\documentclass[a4paper]{amsart}

\title[Existence of L\'evy term structure models]{Existence of L\'evy term structure models}
\author{Damir Filipovi\'c \and Stefan Tappe}
\address{Department of Mathematics\\University of Munich\\Theresienstrasse 39, 80333 Munich, Germany }
\thanks{We are grateful to Bohdan Maslowski, Barbara R\"udiger, Josef Teichmann and Jerzy Zabczyk for their helpful remarks and discussions, and an anonymous referee for bringing our attention to the results of van Gaans \cite{Onno,Onno-Levy}.}

\usepackage{amscd}
\usepackage{amsmath}
\usepackage{amssymb}
\usepackage{amsthm}
\usepackage{bbm}
\usepackage{stmaryrd}
\usepackage{eucal}

\newif\ifpdf
\ifx\pdfoutput\undefined
   \pdffalse        
\else
   \pdfoutput=1     
   \pdftrue
\fi

\ifpdf
   \usepackage[pdftex]{graphicx}
\else
   \usepackage{graphicx}
\fi

\numberwithin{equation}{section}
\swapnumbers

\newtheorem{satz}{Satz}[section]

\newtheorem{theorem}[satz]{Theorem}
\newtheorem{proposition}[satz]{Proposition}
\newtheorem{corollary}[satz]{Corollary}
\newtheorem{lemma}[satz]{Lemma}

\newtheorem{definition}[satz]{Definition}

\newtheorem{remark}[satz]{Remark}

\newtheorem{example}[satz]{Example}

\begin{document}

\begin{abstract}
L\'evy driven term structure models have become an important subject in the mathematical finance literature. This paper provides a comprehensive analysis of the L\'evy driven Heath--Jarrow--Morton type term structure equation. This includes a full proof of existence and uniqueness in particular, which seems to have been lacking in the finance literature so far. 
\bigskip

\end{abstract}


\maketitle\thispagestyle{empty}

\section{Introduction}\label{sec-intro}

A zero coupon bond with maturity $T$ is a financial asset which pays the holder one unit of cash at $T$. Its price at $t\le T$ can be written as
\[ P(t,T)=\exp\left(-\int_t^T f(t,u)\,du\right) \]
where $f(t,T)$ is the forward rate for date $T$. The classical continuous framework for the evolution of the forward rates goes back to Heath, Jarrow and Morton (HJM) \cite{HJM}. They assume that, under the risk-neutral measure, for every date $T$, the forward rates $f(t,T)$ follow an It\^o process of the form
\begin{align}\label{hjm-forward-rates-wiener}
df(t,T) = \bigg( \sum_{i=1}^n \sigma_i(t,T)\int_t^T\sigma_i(t,s)\,ds \bigg) dt + \sum_{i=1}^n \sigma_i(t,T) dW_t^i,\quad t\in [0,T],
\end{align}
where $W=(W^1,\ldots,W^n)$ is a standard Brownian motion in $\mathbb{R}^n$.
The dynamics \eqref{hjm-forward-rates-wiener} guarantee that the discounted zero coupon bond price processes
\[ e^{-\int_0^t f(s,s)\,ds}P(t,T),\quad t\in [0,T],\]
are local martingales for all maturities $T$. This is the well known condition for the absence of arbitrage in the bond market model.

Empirical studies have revealed that models based on Brownian motion only provide a poor fit to observed market data. We refer to \cite[Chap. 5]{Raible}, where it is argued that empirically observed log returns of zero coupon bonds are not normally distributed, a fact, which has long before been known for the distributions of stock returns. Bj\"ork et al.\ \cite{BKR,BKR0}, Eberlein et al.\ \cite{Eberlein-Raible,Eberlein_O, Eberlein_J, Eberlein_K1, Eberlein_K2,Eberlein_Kluge_Review}  and others (\cite{Shirakawa, Jarrow_Madan, Hyll}) thus proposed to replace the classical Brownian motion $W$ in \eqref{hjm-forward-rates-wiener} by a more general process $X = (X^1,\ldots,X^n)$, also taking into account the occurrence of jumps. If $X$ is a L\'evy process, this leads to
\begin{align}\label{hjm-forward-rates}
df(t,T) = \alpha_{\rm HJM}(t,T)dt + \sum_{i=1}^n \sigma_i(t,T) dX_t^i,\quad t\in [0,T].
\end{align}
The HJM drift in \eqref{hjm-forward-rates-wiener} accordingly is replaced by some appropriate $\alpha_{\rm HJM}(t,T)$, which is determined by $\sigma(t,T)$ and the cumulant generating function of $X$, see (\ref{hjm-drift}) below.

Equation \eqref{hjm-forward-rates} constitutes a generic description of the forward rate process $f(t,T,\omega)$ in terms of a stochastic volatility process $\sigma(t,T,\omega)$. From a financial modelling point of view one would rather consider $\sigma(t,T,\omega)$, and thus $\alpha_{\rm HJM}(t,T,\omega)$, to be a function of the prevailing forward curve $T\mapsto f(t-,T,\omega)=\lim_{s\uparrow t}f(s,T,\omega)$, that is
\[ \sigma(t,T,\omega)=\sigma(t,T,f(t-,\cdot,\omega)),\quad \alpha_{\rm HJM}(t,T,\omega)=\alpha_{\rm HJM}(t,T,f(t,\cdot,\omega)).\]
This makes $f(t,T)$ being a {\emph{solution}} of the stochastic equation
\begin{align}\label{hjm-forward-rates2}
\left\{
\begin{array}{rcl}
df(t,T)  & = & \alpha_{\rm HJM}(t,T,f(t,\cdot))dt + \sum_{i=1}^n \sigma_i(t,T,f(t-,\cdot)) dX_t^i,\quad t\in [0,T], \medskip
\\f(0,T) & = & h_0(T)
\end{array}
\right.
\end{align}
for some given initial forward curve $h_0(T)$.

{\emph{Term structure models}} of the type (\ref{hjm-forward-rates2}) are frequently considered in the literature. The typical assumption is that drift $\alpha_{\rm HJM}$ and volatility $\sigma$ depend on the current state of the short rate, $\sigma(t,T,\omega)=\sigma(t,T,f(t-,t,\omega))$, as in \cite{Jeffrey}, \cite{Ritchken}, \cite{Bhar}, \cite{Inui} and \cite{Hyll} (the latter studies models driven by jump-diffusions). A model, where the volatility $\sigma$ is allowed to depend on a finite number of benchmark forward rates, is considered in \cite{Kwon_2001} and \cite{Kwon_2003}. We emphasize that these papers, whose setups are special cases of our present framework, \textit{assume} that the forward rates $f(t,T)$ evolve according to an equation of the kind (\ref{hjm-forward-rates2}). To our knowledge, there has not been yet an explicit proof for the {\emph{existence}} of a solution to \eqref{hjm-forward-rates2} in the mathematical finance literature. We thus provide such a proof in our paper (Theorem~\ref{thm-main-on-hw-spaces} and Corollary~\ref{cor-main-on-hw-spaces}).

Note that \eqref{hjm-forward-rates2} is an infinite-dimensional and therefore non-trivial problem. In fact, \eqref{hjm-forward-rates2} is not simply a system of infinitely many univariate stochastic equations for $f(t,T)$, $t\in [0,T]$, indexed by $T$. Indeed, these equations are coupled as $\alpha_{\rm HJM}$ and $\sigma$ depend on the entire forward curve $f(t-,\cdot)$, say e.g.\ on the short rate $f(t-,t)$, which is a functional of $f(t-,\cdot)$. To express this functional dependence, one switches best to the alternative parametrization
\begin{align*}
r_t(x) = f(t,t+x),\quad x\ge 0,
\end{align*}
which is due to Musiela \cite{Musiela}. We then write $S_t f(x):=f(x+t)$ for the shift operator $S_t$. Equation \eqref{hjm-forward-rates2} becomes in integrated form
\begin{equation}\label{hjmmild}
r_t(x) = S_t h_0(x) + \int_0^t S_{t-s} \alpha_{\rm HJM}(s,s+x,r_s)\, ds + \sum_{i=1}^n \int_0^t S_{t-s} \sigma_i(s,s+x,r_{s-})\, dX_s^i,
\end{equation} 
where $S_{t-s}$ operates on the functions $x\mapsto\alpha_{\rm HJM}(s,s+x,r_s)$ and $x\mapsto\sigma_i(s,s+x,r_{s-})$. Hence, in the spirit of Da~Prato and Zabczyk \cite{Da_Prato}, the process $r_t$ is a so called {\emph{mild solution}} of the stochastic differential equation
\begin{align}\label{hjmX}
\left\{
\begin{array}{rcl}
dr_t & = & \left(\frac{d}{dx} r_t + \alpha_{\rm HJM}(t,r_t)\right)dt + \sum_{i=1}^n \sigma_i(t,r_{t-})\,dX^i_t, \medskip 
\\r_0 & = & h_0
\end{array}
\right.
\end{align}
in some appropriate Hilbert space $H$ of forward curves, where $\frac{d}{dx}$ becomes the generator of the strongly continuous semigroup of shifts $S_t$. Note the slight abuse of notation $\alpha_{\rm HJM}(t,t+\cdot,r) \rightsquigarrow \alpha_{\rm HJM}(t,r)$ and $\sigma_i(t,t+\cdot,r) \rightsquigarrow \sigma_i(t,r)$. In the sequel, we are therefore concerned with the L\'evy HJMM (Heath--Jarrow--Morton--Musiela) equation \eqref{hjmX} in various choices of the state space $H$.

Several authors have dealt with the existence issue for \eqref{hjmX} for the Brownian motion case $X=W$. Bj\"ork and Svensson \cite{Bj_Sv} chose the state space $H_{\beta,\gamma}$ small enough such that $\frac{d}{dx}:H_{\beta,\gamma} \to H_{\beta,\gamma}$ becomes a bounded linear operator. In this case, the methods from finite dimension essentially carry over to \eqref{hjmX}. It turns out, however, that the Bj\"ork--Svensson space $H_{\beta,\gamma}$ is too small and does not contain some important classical term structure models (see \cite{Filipovic}). In \cite{fillnm}, we thus analyzed and solved \eqref{hjmX} for $X=W$ on a larger space $H_w$, where $\frac{d}{dx}$ becomes unbounded. 

In this paper we provide the existence proof for \eqref{hjmX} for the L\'evy case. We proceed as follows. Using an existence result for general Hilbert space valued stochastic differential equations from the appendix, we first show existence for \eqref{hjmX} in the Bj\"ork--Svensson space $H_{\beta,\gamma}$. However, often it turns out that $H_{\beta,\gamma}$ is too small to assert that $\alpha_{\rm HJM}$ lies in $H_{\beta,\gamma}$, even for the very simple case where $\sigma$ is constant and the driver $X$ is a compound Poisson process (Example~\ref{expoiss}). Afterwards, we thus consider \eqref{hjmX} in the larger state space $H_w$ from \cite{fillnm} where $\frac{d}{dx}$ becomes unbounded.

Term structure models based on {\emph{infinite dimensional}} driving processes $X$ are discussed e.g. in \cite{Zabcyk_F2} and \cite{Schmidt} for the L\'evy case. Again, in these papers it is typically assumed that the forward curve evolution satisfies a stochastic differential equation, but the authors do not treat existence and uniqueness of solutions.

The remainder of the paper is organized as follows. In Section \ref{sec-HJM-drift} we introduce some notation and specify the HJM drift $\alpha_{\rm HJM}$, which ensures that the bond market is free of arbitrage. In Section \ref{sec-strong} we treat the existence of strong solutions to \eqref{hjmX} on the Bj\"ork--Svensson space $H_{\beta,\gamma}$. Afterwards, Section \ref{sec-mild} is devoted to the existence of mild and weak solutions to \eqref{hjmX} on the larger space $H_w$ where $\frac{d}{dx}$ becomes unbounded. Section \ref{sec-conclusion} concludes.

For our results of Section \ref{sec-strong} and Section \ref{sec-mild} we apply an existence result for general Hilbert space valued stochastic differential equations, which is derived in the appendix. The ground for this result, Theorem \ref{thm-sde-solution}, is prepared by two works of van Gaans \cite{Onno,Onno-Levy}. In addition to his result \cite[Thm. 4.1]{Onno-Levy} we prove that the mild solution to \eqref{hjmX} has a c\`adl\`ag modification, and that there exists a unique weak solution.

The c\`adl\`ag property of the solution is an important feature for financial applications. Indeed, general arbitrage theory \cite{ds94} requires that the basic financial instruments, here the implied zero coupon bond prices $P(t,T)=\exp(-\int_0^{T-t} r_t(x)\,dx)$, are real semimartingales and therefore have c\`adl\`ag paths. This essentially requires c\`adl\`ag paths of the weak solution $(r_t)$, which is satisfied in our framework.

As it turns out, the stochastic integral of van Gaans \cite{Onno-Levy} is not consistent with the usual It\^o-integral, which is used for financial modelling. Therefore, after giving an overview and the required notation in Appendix \ref{app-overview}, we show in Appendix \ref{app-integration} that the stochastic integral of van Gaans always has a c\`adl\`ag modification and analyze when it coincides with the It\^o-integral. Then, in Appendix \ref{app-SDE}, we prove Theorem \ref{thm-sde-solution}, the existence and uniqueness result for Hilbert space valued stochastic equations. At the end of Appendix \ref{app-SDE} we give an overview of related literature.

\section{The HJM drift condition}\label{sec-HJM-drift}

Throughout this text, $X^1,\ldots,X^n$ denote independent real-valued L\'evy processes on a filtered probability space $(\Omega,\mathcal{F},(\mathcal{F}_t)_{t \geq 0},\mathbb{P})$ satisfying the usual conditions.

Let $H$ be a separable Hilbert space representing the space of forward curves and let
$\sigma_1,\ldots,\sigma_n : \mathbb{R}_+ \times H \rightarrow H$ be the volatilities.
Recall that a L\'evy term structure model of the form (\ref{hjmX}) is free of arbitrage if the probability measure $\mathbb{P}$ is a local martingale measure, that is all discounted bond prices are local martingales.

In order to provide a condition which ensures that $\mathbb{P}$ is a local martingale measure, we assume that there are compact intervals $[a_1,b_1],\ldots,[a_n,b_n]$ having zero as an inner point such that the L\'evy measures $F_1,\ldots,F_n$ of $X^1,\ldots,X^n$, respectively, satisfy for $i = 1,\ldots,n$
\begin{align}\label{canonical-cond0}
\int_{|x| > 1} e^{zx}F_i(dx) < \infty, \quad z \in [a_i,b_i].
\end{align}
Condition (\ref{canonical-cond0}) ensures that the cumulant generating functions
\begin{align}\label{cumulant}
\Psi_i(z):=\ln\mathbb{E}[\exp(z X_1^i)], \quad i = 1,\ldots,n 
\end{align}
exist on $[a_i,b_i]$ and that they are of class $C^{\infty}$ (see \cite[Lemma 26.4]{Sato}). Moreover, the L\'evy processes $X^i$ possess moments of arbitrary order. Let $[c_i,d_i] \subset (a_i,b_i)$ be further compact intervals having zero as an inner point.

For any continuous function $h : \mathbb{R}_+ \rightarrow \mathbb{R}$ we define $Th : \mathbb{R}_+ \rightarrow \mathbb{R}$ as
\begin{align}\label{def-of-T}
Th(x) := \int_0^x h(\eta)d \eta.
\end{align}
For $i=1,\ldots,n$ denote
\begin{align*}
A_H^{\Psi_i} := \{ h \in H : -Th(\mathbb{R}_+) \subset [c_i,d_i] \}.
\end{align*}
Provided $\sigma_i(\mathbb{R}_+ \times H) \subset A_H^{\Psi_i}$ for $i = 1,\ldots,n$, the \textit{HJM drift}
\begin{equation}\label{hjm-drift}
\begin{aligned}
\alpha_{\rm HJM}(t,r)(x) &= \sum_{i=1}^n \frac{d}{dx} \Psi_i \left( -\int_0^x \sigma_i(t,r)(\eta) d \eta \right) 
\\  &= - \sum_{i=1}^n \sigma_i(t,r)(x) \Psi_i' \left( -\int_0^x \sigma_i(t,r)(\eta) d \eta \right)
\end{aligned} 
\end{equation}
is well defined pointwise for all $x$. The HJM drift condition \eqref{hjm-drift} implies that $\mathbb{P}$ is a local martingale measure. It is derived in \cite[Sec. 2.1]{Eberlein_O} for the present L\'evy case, using the results of the more general setup in \cite{BKR}. For an analogous drift condition in the infinite dimensional L\'evy setting, see \cite{Zabcyk_F2}.


The HJM drift specification \eqref{hjm-drift} causes some problems for an immediate application of Theorem \ref{thm-sde-solution}. First of all, we have to ensure that $\alpha_{\rm HJM}(t,r) \in H$ for all $(t,r) \in \mathbb{R}_+ \times H$. Furthermore, we have to establish for an application of Theorem \ref{thm-sde-solution} that for Lipschitz functions $\sigma_1,\ldots,\sigma_n$ the drift $\alpha_{\rm HJM}$ is again a Lipschitz function.

These demandings emphasize that we have to be careful about the choice of the space $H$ of forward curves. Another desirable feature of $H$ is that for every $x \in \mathbb{R}_+$ the point evaluation $h \mapsto h(x) : H \rightarrow \mathbb{R}$ is a continuous linear functional. Because then the variation of constants formula (\ref{hjmmild}) is satisfied for all $x \in \mathbb{R}_+$, whenever $(r_t)$ is a mild solution of (\ref{hjmX}).

In the upcoming Section \ref{sec-strong} we deal with the existence of strong solutions to (\ref{hjmX}), and Section \ref{sec-mild} is devoted to the existence of mild and weak solutions to (\ref{hjmX}).

\section{Forward curve evolutions as strong solutions of infinite dimensional stochastic differential equations}\label{sec-strong}

In this section, where we deal with the existence of strong solutions to (\ref{hjmX}), we consider the spaces $H_{\beta,\gamma}$ of forward curves, which have been used by Bj\"ork and Svensson in \cite{Bj_Sv}.

We fix real numbers $\beta > 1$ and $\gamma > 0$. Let $H_{\beta,\gamma}$ be the linear space of all $h \in C^{\infty}(\mathbb{R}_+;\mathbb{R})$ satisfying
\begin{align*}
\sum_{n=0}^{\infty} \left( \frac{1}{\beta} \right)^n \int_0^{\infty} \left(\frac{d^n h(x)}{dx^n}\right)^2 e^{-\gamma x} dx < \infty,
\end{align*}
We define the inner product
\begin{align*}
\langle g,h \rangle_{\beta,\gamma} := \sum_{n=0}^{\infty} \left( \frac{1}{\beta} \right)^n \int_0^{\infty} \left(\frac{d^n g(x)}{dx^n}\right) \left(\frac{d^n h(x)}{dx^n}\right) e^{-\gamma x} dx
\end{align*}
and denote the corresponding norm by $\| \cdot \|_{\beta,\gamma}$.

\begin{proposition}
The space $(H_{\beta,\gamma},\langle \cdot,\cdot \rangle_{\beta,\gamma})$ is a separable Hilbert space and for each $x \in \mathbb{R}_+$, the point evaluation $h \mapsto h(x) : H_{\beta,\gamma} \rightarrow \mathbb{R}$ is a continuous linear functional.
\end{proposition}

\begin{proof}
This is a consequence of \cite[Prop. 4.2]{Bj_Sv}.
\end{proof}

The fact that each point evaluation is a continuous linear functional ensures that forward curves $(r_t)$ solving (\ref{hjmX}) satisfy the variation of constants formula (\ref{hjmmild}).


\begin{proposition}\label{A-bounded-operator}
We have $\frac{d}{dx} \in \mathcal{L}(H_{\beta,\gamma})$, i.e. $\frac{d}{dx}$ is a bounded linear operator on $H_{\beta,\gamma}$.
\end{proposition}

\begin{proof}
The assertion is a consequence of \cite[Prop. 4.2]{Bj_Sv}.
\end{proof}


\begin{theorem}\label{thm-main-on-bs-spaces}
Let $\sigma_i : \mathbb{R}_+ \times H_{\beta,\gamma} \rightarrow H_{\beta,\gamma}$ be continuous and satisfying $\sigma_i(\mathbb{R}_+ \times H_{\beta,\gamma}) \subset A_{H_{\beta,\gamma}}^{\Psi_i}$ for $i = 1,\ldots,n$. Assume that $\alpha_{\rm HJM}(t,r) \in H_{\beta,\gamma}$ for all $(t,r) \in \mathbb{R}_+ \times H_{\beta,\gamma}$. Furthermore, assume that $\alpha_{\rm HJM}(t,r) : \mathbb{R}_+ \times H_{\beta,\gamma} \rightarrow H_{\beta,\gamma}$ is continuous and that
there is a constant $L \geq 0$ such that for all $t \in \mathbb{R}_+$ and $h_1,h_2 \in H_{\beta,\gamma}$ we have
\begin{align*}
\| \alpha_{\rm HJM}(t,h_1) - \alpha_{\rm HJM}(t,h_2) \|_{\beta,\gamma} &\leq L \| h_1 - h_2 \|_{\beta,\gamma},
\\ \| \sigma_i(t,h_1) - \sigma_i(t,h_2) \|_{\beta,\gamma} &\leq L \| h_1 - h_2 \|_{\beta,\gamma}, \quad i=1,\ldots,n.
\end{align*}
Then, for each $h_0 \in H_{\beta,\gamma}$, there exists a unique strong adapted c\`adl\`ag solution $(r_t)_{t \geq 0}$ to (\ref{hjmX}) with $r_0 = h_0$ satisfying
\begin{align}\label{sup-finite-in-BS-spaces}
\mathbb{E} \bigg[ \sup_{t \in [0,T]} \| r_t \|_{\beta,\gamma}^2 \bigg] < \infty \quad \text{for all $T > 0$.}
\end{align}
\end{theorem}

\begin{proof}
Taking into account Proposition \ref{A-bounded-operator}, the result is a consequence of Corollary \ref{cor-solution-of-SDE}.
\end{proof}

Unfortunately, Theorem \ref{thm-main-on-bs-spaces} has some shortcomings, namely it is demanded that the drift term $\alpha_{\rm HJM}$ according to the HJM drift condition (\ref{hjm-drift}) maps again into the space $H_{\beta,\gamma}$. The following simple counter example shows that this condition may be violated.

\begin{example}\label{expoiss}
Let $\sigma = -1$ and $X$ be a compound Poisson process with intensity $\lambda = 1$ and jump size distribution $N(0,1)$. Notice that the compound Poisson process satisfies the exponential moments condition (\ref{canonical-cond0}) for all $z \in \mathbb{R}$, because its 
L\'evy measure is given by
\begin{align*}
F(dx) = \frac{1}{\sqrt{2 \pi}} e^{- \frac{x^2}{2}} dx.
\end{align*}
But we have $\alpha_{\rm HJM} \notin H_{\beta,\gamma}$, because
\begin{align*}
&\int_0^{\infty} \alpha_{\rm HJM}(x)^2 e^{-\gamma x} dx = \int_0^{\infty} \left( \frac{d}{dx} \Psi(x) \right)^2 e^{-\gamma x} dx
\\ &= \int_0^{\infty} \left( \frac{d}{dx} \left( e^{\frac{x^2}{2}} - 1 \right) \right)^2 e^{-\gamma x} dx = \int_0^{\infty} x^2 e^{x^2 - \gamma x} dx = \infty.
\end{align*}
\end{example}

The phenomena that the drift $\alpha_{\rm HJM}$ may be located outside the space of forward curves $H_{\beta,\gamma}$ has to do with the fact that the space $H_{\beta,\gamma}$ is a very small space in a sense, in particular, every function must necessarily be real-analytic (see \cite[Prop. 4.2]{Bj_Sv}). 

The small size of this space arises from the requirement that $\frac{d}{dx}$ should be a bounded operator, because we deal with the existence of \textit{strong} solutions. When dealing with \textit{mild} and \textit{weak} solutions in the next Section \ref{sec-mild}, problems of this kind will not occur.
 
Nevertheless, for certain types of term structure models, we can apply Theorem \ref{thm-main-on-bs-spaces}. For this purpose, we proceed with a lemma. For a given real-analytic function $h : \mathbb{R}_+ \rightarrow \mathbb{R}$ it is, in general, difficult to decide whether $h$ belongs to $H_{\beta,\gamma}$ or not. For the following functions this can be provided.

\begin{lemma}\label{lemma-in-space}
Every polynomial $p$ belongs to $H_{\beta,\gamma}$, and for $\delta \in \mathbb{R}$ satisfying $\delta^2 < \beta$ and $\delta < \frac{\gamma}{2}$, the function $h(x) = e^{\delta x}$ belongs to $H_{\beta,\gamma}$.
\end{lemma}

\begin{proof}
The first statement is clear. For $h(x) = e^{\delta x}$ we obtain
\begin{align*}
&\sum_{n=0}^{\infty} \left( \frac{1}{\beta} \right)^n \int_0^{\infty} \left(\frac{d^n h(x)}{dx^n}\right)^2 e^{-\gamma x} dx = \sum_{n=0}^{\infty} \left( \frac{1}{\beta} \right)^n \int_0^{\infty} \left( \delta^n e^{\delta x} \right)^2 e^{-\gamma x} dx
\\ &= \sum_{n=0}^{\infty} \left( \frac{\delta^2}{\beta} \right)^n \int_0^{\infty} e^{-(\gamma - 2\delta)x} dx = \frac{1}{1 - \frac{\delta^2}{\beta}} \cdot \frac{1}{\gamma - 2\delta} = \frac{\beta}{(\beta - \delta^2)(\gamma - 2\delta)},
\end{align*}
whence $h \in H_{\beta,\gamma}$.
\end{proof}

Let $n=3$, that is we have three independent driving processes. We denote by $X^1,X^2$ two standard Wiener processes, and $X^3$ is a Poisson process with intensity $\lambda > 0$.
We specify the volatilities as
\begin{align*}
\sigma_1(r)(x) = \varphi_1(r) p(x), \text{ } \sigma_2(r)(x) = \varphi_2(r) e^{\delta x} \text{ and } \sigma_3(r)(x) = -\eta,
\end{align*}
where $p$ is a polynomial, $\delta, \eta \in \mathbb{R}$ satisfy $4\delta^2 < \beta$, $\delta < \frac{\gamma}{4}$ and $\eta^2 < \beta$, $\eta < \frac{\gamma}{2}$, and where $\varphi_i : H_{\gamma,\beta} \rightarrow \mathbb{R}$ for $i=1,2$. Note that $\sigma_i(H_{\beta,\gamma}) \subset H_{\beta,\gamma}$ for $i=1,2,3$ by Lemma \ref{lemma-in-space}. The drift according to the HJM drift condition (\ref{hjm-drift}) is given by
\begin{align*}
\alpha_{\rm HJM}(r)(x) = \frac{d}{dx} \left[ \frac{1}{2} \varphi_1(r)^2 q(x)^2 + \frac{1}{2} \varphi_2(r)^2 \left( \frac{e^{\delta x} - 1}{\delta} \right)^2 + \lambda \left( e^{\eta x} - 1 \right) \right],
\end{align*}
where $q(x) = \int_0^x p(\eta) d\eta$ is again a polynomial. From Lemma \ref{lemma-in-space} and Proposition \ref{A-bounded-operator} we infer $\alpha_{\rm HJM}(H_{\beta,\gamma}) \subset H_{\beta,\gamma}$.

\begin{proposition}\label{prop-ex-1}
Assume there is a constant $L \geq 0$ such that for all $h_1,h_2 \in H_{\beta,\gamma}$ we have
\begin{align*}
|\varphi_i(h_1) - \varphi_i(h_2)| &\leq L \| h_1 - h_1 \|_{\beta,\gamma}, \quad i=1,2
\\ |\varphi_i(h_1)^2 - \varphi_i(h_2)^2| &\leq L \| h_1 - h_1 \|_{\beta,\gamma}, \quad i=1,2.
\end{align*}
Then, for each $h_0 \in H_{\beta,\gamma}$, there exists a unique strong adapted c\`adl\`ag solution $(r_t)_{t \geq 0}$ to (\ref{hjmX}) with $r_0 = h_0$ satisfying (\ref{sup-finite-in-BS-spaces}).
\end{proposition}

\begin{proof}
We have for all $h_1,h_2 \in H_{\beta,\gamma}$
\begin{align*}
\| \sigma_1(h_1) - \sigma_1(h_2) \| &\leq L \| p \|_{\beta,\gamma} \| h_1 - h_2 \|_{\beta,\gamma},
\\ \| \sigma_2(h_1) - \sigma_2(h_2) \| &\leq L \| e^{\delta \bullet} \|_{\beta,\gamma} \| h_1 - h_2 \|_{\beta,\gamma}. 
\end{align*}
Using Proposition \ref{A-bounded-operator}, we obtain for all $h_1,h_2 \in H_{\beta,\gamma}$
\begin{align*}
\| \alpha_{\rm HJM}(h_1) - \alpha_{\rm HJM}(h_2) \|_{\beta,\gamma} \leq \frac{L}{2} \| A \|_{\mathcal{L}(H_{\beta,\gamma})} \left( \| q^2 \|_{\beta,\gamma} + \| \textstyle\frac{1}{\delta^2} (e^{\delta \bullet} - 1)^2 \|_{\beta,\gamma} \right) \| h_1 - h_2 \|_{\beta,\gamma}.
\end{align*}
Applying Theorem \ref{thm-main-on-bs-spaces} completes the proof.
\end{proof}

In order to generalize Proposition \ref{prop-ex-1}, by allowing that $\eta$ may depend on the present state of the forward curve, we prepare two auxiliary results.

\begin{lemma}\label{lemma-part-int-3-factors}
Let $\gamma > 0$ and $g,h \in C^1(\mathbb{R}_+;\mathbb{R})$. Assume there are $c > 0$, $\varepsilon \in (-\infty,\gamma)$ and $x_0 \in \mathbb{R}_+$ such that
\begin{align*}
|g(x) h(x)| \leq c e^{\varepsilon x} \quad \text{for all $x \geq x_0$.}
\end{align*}
Then we have
\begin{align*}
&\int_0^{\infty} g(x) h(x) e^{-\gamma x} dx 
\\ &= \frac{1}{\gamma} \left[ g(0)h(0) + \int_0^{\infty} g'(x) h(x) e^{-\gamma x} dx
+ \int_0^{\infty} g(x) h'(x) e^{-\gamma x} dx \right].
\end{align*}
\end{lemma}

\begin{proof}
Performing partial integration with three factors, we obtain
\begin{align*}
\Big[ g(x) h(x) e^{-\gamma x} \Big]_0^{\infty} &= \int_0^{\infty} g'(x) h(x) e^{-\gamma x} dx + \int_0^{\infty} g(x) h'(x) e^{-\gamma x} dx 
\\ & \quad - \gamma \int_0^{\infty} g(x) h(x) e^{-\gamma x} dx.
\end{align*}
By hypothesis, we have $\lim_{x \rightarrow \infty} g(x) h(x) e^{-\gamma x} = 0$, and
so the stated formula follows.
\end{proof}

\begin{lemma}\label{lemma-estimation-fder-f}
Let $\gamma > 0$ and $h \in C^2(\mathbb{R}_+;\mathbb{R})$ be such that $h, h', h'' \geq 0$. Assume there are $c > 0$, $\varepsilon \in (-\infty,\frac{\gamma}{2})$ and $x_0 \in \mathbb{R}_+$ such that
\begin{align*}
|h(x)| \leq c e^{\varepsilon x} \text{ and } |h'(x)| \leq c e^{\varepsilon x} \quad \text{for all $x \geq x_0$.}
\end{align*}
Then we have
\begin{align*}
\int_0^{\infty} h'(x)^2 e^{-\gamma x} dx \leq \frac{\gamma^2}{2} \int_0^{\infty} h(x)^2 e^{-\gamma x} dx.
\end{align*}
\end{lemma}

\begin{proof}
Using two times Lemma \ref{lemma-part-int-3-factors}, we obtain
\begin{align*}
&\int_0^{\infty} h(x)^2 e^{-\gamma x} dx 
= \frac{2}{\gamma} \int_0^{\infty} h(x) h'(x) e^{-\gamma x} dx + \frac{1}{\gamma} h(0)^2
\\ &= \frac{2}{\gamma^2} \left[ \int_0^{\infty} h'(x)^2 e^{-\gamma x} dx +
\int_0^{\infty} h(x) h''(x) e^{-\gamma x} dx \right] + \frac{1}{\gamma} \left[ h(0)^2 + \frac{2}{\gamma} h(0) h'(0) \right].
\end{align*}
Since $h, h', h'' \geq 0$ by hypothesis, the stated inequality follows.
\end{proof}

Now we generalize Proposition \ref{prop-ex-1} by assuming that $\eta : H_{\beta,\gamma} \rightarrow \mathbb{R}$ is allowed to depend on the current state of the forward curve. The rest of our present framework is exactly as in Proposition \ref{prop-ex-1}.

\begin{proposition}\label{prop-ex-2}
Assume that, in addition to the hypothesis of Proposition \ref{prop-ex-1}, we have $\gamma \leq \sqrt{2}$, $\eta(H_{\beta,\gamma}) \subset [0,\frac{\gamma}{2}) \cap [0,\sqrt{\beta})$ and
\begin{align*} 
| \eta(h_1) - \eta(h_2) | &\leq L \| h_1 - h_2 \|_{\beta,\gamma}
\end{align*}
for all $h_1,h_2 \in H_{\beta,\gamma}$. Then, for each $h_0 \in H_{\beta,\gamma}$, there exists a unique strong adapted c\`adl\`ag solution $(r_t)_{t \geq 0}$ to (\ref{hjmX}) with $r_0 = h_0$ satisfying (\ref{sup-finite-in-BS-spaces}).
\end{proposition}

\begin{proof}
It suffices to show that $\Gamma : H_{\beta,\gamma} \rightarrow H_{\beta,\gamma}$ defined as $\Gamma(r)(x) := e^{\eta(r) x}$ is Lipschitz continuous. So let $h_1,h_2 \in H_{\beta,\gamma}$ be arbitrary. Without loss of generality we assume that $\eta(h_2) \leq \eta(h_1)$. Observe that all derivatives of $\Gamma(h_1) - \Gamma(h_2)$ are non-negative. So we obtain by applying Lemma \ref{lemma-estimation-fder-f} (notice that $\gamma \leq \sqrt{2}$ by hypothesis), and the Lipschitz property $|e^x - e^y| \leq e^x |x-y|$ for $y \leq x$ that
\begin{align*}
\left\| \Gamma(h_1) - \Gamma(h_2) \right\|_{\beta,\gamma}^2 &= \sum_{n=0}^{\infty} \left( \frac{1}{\beta} \right)^n \int_0^{\infty} \left( \eta(h_1)^n e^{\eta(h_1)x} - \eta(h_2)^n e^{\eta(h_2)x} \right)^2 e^{-\gamma x} dx 
\\ &\leq \frac{\beta}{\beta - 1} \int_0^{\infty} \left( e^{\eta(h_1)x} - e^{\eta(h_2)x} \right)^2 e^{-\gamma x} dx
\\ &\leq \frac{\beta}{\beta - 1} \int_0^{\infty} \left( e^{\eta(h_1)x} (\eta(h_1) - \eta(h_2))x \right)^2 e^{-\gamma x} dx
\\ &\leq \frac{\beta}{\beta - 1} \left( \int_0^{\infty} \left( x e^{\eta(h_1)x} \right)^2 e^{-\gamma x} dx \right) L^2 \| h_1 - h_2 \|_{\beta,\gamma}^2.
\end{align*}
The integral is finite, because we have $\eta(h_1) \in [0,\frac{\gamma}{2})$ by assumption. Applying Theorem \ref{thm-main-on-bs-spaces} finishes the proof.
\end{proof}

\section{Forward curve evolutions as mild and weak solutions of infinite dimensional stochastic differential equations}\label{sec-mild}

In this section, where we deal with the existence of mild and weak solutions to (\ref{hjmX}), we consider the spaces $H_w$ of forward curves, which have been introduced in \cite[Chap. 5]{fillnm}.

Let $w : \mathbb{R}_+ \rightarrow [1,\infty)$ be a non-decreasing $C^1$-function such that
$w^{-\frac{1}{3}} \in L^1(\mathbb{R}_+)$.

\begin{example}\label{ex-weighting-function}
$w(x) = e^{\alpha x}$, for $\alpha > 0$.
\end{example}

\begin{example}
$w(x) = (1+x)^{\alpha}$, for $\alpha > 3$.
\end{example}

Let $H_w$ be the linear space of all absolutely continuous functions $h : \mathbb{R}_+ \rightarrow \mathbb{R}$ satisfying
\begin{align*}
\int_{\mathbb{R}_+} |h'(x)|^2 w(x) dx < \infty,
\end{align*}
where $h'$ denotes the weak derivative of $h$. We define the inner product
\begin{align*}
( g,h )_w := g(0)h(0) + \int_{\mathbb{R}_+} g'(x)h'(x)w(x)dx
\end{align*}
and denote the corresponding norm by $\interleave \cdot \interleave_w$. Since forward curves flatten for large time to maturity $x$, the choice of $H_w$ is reasonable from an economic point of view.

\begin{proposition}
The space $(H_w,( \cdot,\cdot )_w)$ is a separable Hilbert space. Each $h \in H_w$ is continuous, bounded and the limit $h(\infty) := \lim_{x \rightarrow \infty} h(x)$ exists. Moreover, for each $x \in \mathbb{R}_+$, the point evaluation $h \mapsto h(x) : H_w \rightarrow \mathbb{R}$ is a continuous linear functional.
\end{proposition}

\begin{proof}
All of these statements can be found in the proof of \cite[Thm. 5.1.1]{fillnm}.
\end{proof}

The fact that each point evaluation is a continuous linear functional ensures that forward curves $(r_t)$ solving (\ref{hjmX}) satisfy the variation of constants formula (\ref{hjmmild}).

Defining the constants $C_1, \ldots, C_4 > 0$ as
\begin{align*}
C_1 := \| w^{-1} \|_{L^1(\mathbb{R}_+)}^{\frac{1}{2}}, \quad
C_2 := 1 + C_1, \quad
C_3 := \| w^{- \frac{1}{3}} \|_{L^1(\mathbb{R}_+)}^2, \quad
C_4 := \| w^{- \frac{1}{3}} \|_{L^1(\mathbb{R}_+)}^{\frac{7}{2}},
\end{align*}
we have for all $h \in H_w$ the estimates
\begin{align}
\label{estimate-c1} \| h' \|_{L^1(\mathbb{R}_+)} &\leq C_1 \interleave h \interleave_w,
\\ \label{estimate-c2} \| h \|_{L^{\infty}(\mathbb{R}_+)} &\leq C_2 \interleave h \interleave_w,
\\ \label{est-h} \| h - h(\infty) \|_{L^1(\mathbb{R}_+)} &\leq C_3 \interleave h \interleave_w,
\\ \label{est-h-4} \| (h - h(\infty))^4 w \|_{L^1(\mathbb{R}_+)} &\leq C_4 \interleave h \interleave_w^4,
\end{align}
which also follows by inspecting the proof of \cite[Thm. 5.1.1]{fillnm}.

Since for an application of Theorem \ref{thm-sde-solution} we require that the shift semigroup $(S_t)_{t \geq 0}$ defined by $S_t h = h(t + \cdot)$ for $t \in \mathbb{R}_+$ is pseudo-contractive in a closed subspace of $H_w$, we perform an idea, which is due to Tehranchi \cite{Tehranchi}, namely we change to the inner product
\begin{align*}
\langle g,h \rangle_w := g(\infty)h(\infty) + \int_{\mathbb{R}_+} g'(x)h'(x)w(x)dx
\end{align*}
and denote the corresponding norm by $\| \cdot \|_w$. The estimates (\ref{estimate-c1})--(\ref{est-h-4}) are also valid with the norm $\| \cdot \|_w$ for all $h \in H_w$, which is proven exactly as for the original norm $\interleave \cdot \interleave_w$. Therefore we conclude, by using (\ref{estimate-c2}),
\begin{align*}
\frac{1}{(1+C_2^2)^{\frac{1}{2}}} \| h \|_w \leq \interleave h \interleave_w \leq (1 + C_2^2)^{\frac{1}{2}} \| h \|_w, \quad h \in H_w
\end{align*}
showing that $\| \cdot \|_w$ and $\interleave \cdot \interleave_w$ are equivalent norms on $H_w$. From now on, we shall work with the norm $\| \cdot \|_w$.

\begin{proposition}\label{prop-pseudo-contractive}
$(S_t)$ is a $C_0$-semigroup in $H_w$ with generator $\frac{d}{dx} : \mathcal{D}(\frac{d}{dx}) \subset H_w \rightarrow H_w$, $\frac{d}{dx}h = h'$, and domain
 \begin{align*}
\mathcal{D}({\textstyle \frac{d}{dx}}) = \{ h \in H_w \, | \, h' \in H_w \}.
\end{align*}
The subspace $H_w^0 := \{ h \in H_w \, | \, h(\infty) = 0 \}$ is a closed subspace of $H_w$ and $(S_t)$ is contractive in $H_w^0$ with respect to the norm $\| \cdot \|_w$.
\end{proposition}

\begin{proof}
Except for the last statement, we refer to the proof of \cite[Thm. 5.1.1]{fillnm}. By the monotonicity of $w$ we have
\begin{align*}
\| S_t h \|_w^2 = \int_{\mathbb{R}_+} |h'(x+t)|^2w(x)dx \leq  \| h \|_w^2
\end{align*}
for all $t \in \mathbb{R}_+$ and $h \in H_w^0$, showing that $(S_t)$ is contractive in $H_w^0$.
\end{proof}

We define for any $h = (h_1,\ldots,h_n) \in \Pi_{i=1}^n A_{H_w^0}^{\Psi_i}$
\begin{align}\label{def-of-capital-sigma}
\Sigma h(x) := - \sum_{i=1}^n h_i(x) \Psi_i' \left( -\int_0^x h_i(\eta) d\eta \right), \quad x \in \mathbb{R}_+.
\end{align}

\begin{proposition}\label{prop-const-C5}
There is a constant $C_5 > 0$ such that for all $g,h \in \Pi_{i=1}^n A_{H_w^0}^{\Psi_i}$ we have
\begin{align}\label{locally-Lipschitz}
\| \Sigma g - \Sigma h \|_w \leq C_5 \sum_{i=1}^n \left( 1 + \| h_i \|_w + \| g_i \|_w + \| g_i \|_w^2 \right) \| g_i - h_i \|_w.
\end{align}
Furthermore, for each $h \in \Pi_{i=1}^n A_{H_w^0}^{\Psi_i}$ we have $\Sigma h \in H_w^0$, and the map $\Sigma : \Pi_{i=1}^n A_{H_w^0}^{\Psi_i} \rightarrow H_w^0$ is continuous.
\end{proposition}

\begin{proof}
We define 
\begin{align*}
K_i := \sup_{x \in [c_i,d_i]} |\Psi_i'(x)|, \text{ } L_i := \sup_{x \in [c_i,d_i]} |\Psi_i''(x)| \text{ and } M_i := \sup_{x \in [c_i,d_i]} |\Psi_i'''(x)|
\end{align*}
for $i = 1,\ldots,n$. By the boundedness of the derivatives $\Psi_i'$ on $[c_i,d_i]$, the definition (\ref{def-of-capital-sigma}) of $\Sigma$ yields that for each $h \in \Pi_{i=1}^n A_{H_w^0}^{\Psi_i}$ the limit $\Sigma h (\infty) := \lim_{x \rightarrow \infty} \Sigma h(x)$ exists and 
\begin{align}\label{limit-cap-sigma-zero}
\Sigma h (\infty) = 0, \quad h \in \Pi_{i=1}^n A_{H_w^0}^{\Psi_i}.
\end{align}
By using (\ref{limit-cap-sigma-zero}) and the universal inequality
\begin{align*}
|x_1 + \ldots + x_k|^2 \leq k \left( |x_1|^2 + \ldots + |x_k|^2 \right), \quad k \in \mathbb{N}
\end{align*} 
we get for arbitrary $g,h \in \Pi_{i=1}^n A_{H_w^0}^{\Psi_i}$ the estimation
\begin{align*}
&\| \Sigma g - \Sigma h \|_w^2 = \int_{\mathbb{R}_+} \Big| \sum_{i=1}^n h_i'(x) \Psi_i' \left( -\int_0^x h_i(\eta)d\eta \right) - \sum_{i=1}^n g_i'(x) \Psi_i' \left( -\int_0^x g_i(\eta)d\eta \right) 
\\ &\quad + \sum_{i=1}^n g_i(x)^2 \Psi_i'' \left( -\int_0^x g_i(\eta)d\eta \right) - \sum_{i=1}^n h_i(x)^2 \Psi_i'' \left( -\int_0^x h_i(\eta)d\eta \right) \Big|^2 w(x) dx
\\ &\quad \leq 4n(I_1 + I_2 + I_3 + I_4),
\end{align*}
where we have put
\begin{align*}
I_1 &:= \sum_{i=1}^n \int_{\mathbb{R}_+} |h_i'(x)|^2 \Big| \Psi_i' \left( -\int_0^x h_i(\eta)d\eta \right) - \Psi_i' \left( -\int_0^x g_i(\eta)d\eta \right) \Big|^2 w(x) dx,
\\ I_2 &:= \sum_{i=1}^n \int_{\mathbb{R}_+} \Psi_i' \left( -\int_0^x g_i(\eta)d\eta \right)^2 |h_i'(x) - g_i'(x)|^2 w(x) dx,
\\ I_3 &:= \sum_{i=1}^n \int_{\mathbb{R}_+} g_i(x)^4 \left[ \Psi_i'' \left( -\int_0^x g_i(\eta)d\eta \right) - \Psi_i'' \left( -\int_0^x h_i(\eta)d\eta \right) \right]^2 w(x) dx,
\\ I_4 &:= \sum_{i=1}^n \int_{\mathbb{R}_+} \Psi_i'' \left( -\int_0^x h_i(\eta)d\eta \right)^2 (g_i(x)^2 - h_i(x)^2)^2 w(x) dx.
\end{align*}
Using (\ref{est-h}) yields
\begin{align*}
I_1 \leq \sum_{i=1}^n L_i^2 \| h_i \|_w^2 \| g_i-h_i \|_{L^1(\mathbb{R}_+)}^2 \leq C_3^2 \sum_{i=1}^n L_i^2 \| h_i \|_w^2 \| g_i-h_i \|_w^2,
\end{align*}
and $I_2$ is estimated as
\begin{align*}
I_2 \leq \sum_{i=1}^n K_i^2 \| g_i-h_i \|_w^2.
\end{align*}
Taking into account (\ref{est-h}) and (\ref{est-h-4}), we get
\begin{align*}
I_3 \leq \sum_{i=1}^n M_i^2 \| g_i^4 w \|_{L^1(\mathbb{R}_+)} \| g_i-h_i \|_{L^1(\mathbb{R}_+)}^2 \leq C_3^2 C_4 \sum_{i=1}^n M_i^2 \| g_i \|_w^4 \| g_i-h_i \|_w^2,
\end{align*}
and by using H\"older's inequality and (\ref{est-h-4}), we obtain
\begin{align*}
I_4 &\leq \sum_{i=1}^n L_i^2 \int_{\mathbb{R}_+} (g_i(x)+h_i(x))^2 w(x)^{\frac{1}{2}} (g_i(x)-h_i(x))^2 w(x)^{\frac{1}{2}} dx
\\ &\leq \sum_{i=1}^n L_i^2 \|(g_i+h_i)^4 w \|_{L^1(\mathbb{R}_+)}^{\frac{1}{2}} \|(g_i-h_i)^4 w \|_{L^1(\mathbb{R}_+)}^{\frac{1}{2}}
\\ &\leq 2 C_4 \sum_{i=1}^n L_i^2 (\| g_i \|_w^2 + \| h_i \|_w^2) \| g_i-h_i \|_w^2,
\end{align*}
which gives us the desired estimation (\ref{locally-Lipschitz}). For all $h \in \Pi_{i=1}^n A_{H_w^0}^{\Psi_i}$ we have $\Sigma h \in H_w^0$ by (\ref{locally-Lipschitz}) and (\ref{limit-cap-sigma-zero}), and the map $\Sigma : \Pi_{i=1}^n A_{H_w^0}^{\Psi_i} \rightarrow H_w^0$ is locally Lipschitz continuous by (\ref{locally-Lipschitz}).
\end{proof}

By Proposition \ref{prop-const-C5} we can, for given volatilities $\sigma_i : \mathbb{R}_+ \times H_w \rightarrow H_w^0$ satisfying $\sigma_i(\mathbb{R}_+ \times H_w) \subset A_{H_w^0}^{\Psi_i}$ for $i = 1,\ldots,n$, define the drift term $\alpha_{\rm HJM}$ according to the HJM drift condition (\ref{hjm-drift}) by
\begin{align}\label{alpha-HJM-Hw}
\alpha_{\rm HJM} := \Sigma \circ \sigma : \mathbb{R}_+ \times H_w \rightarrow H_w^0,
\end{align}
where $\sigma = (\sigma_1,\ldots,\sigma_n)$.

Now, we are ready to establish the existence of L\'evy term structure models on the space $H_w$ of forward curves.

\begin{theorem}\label{thm-main-on-hw-spaces}
Let $\sigma_i : \mathbb{R}_+ \times H_w \rightarrow H_w^0$ be continuous and satisfying $\sigma_i(\mathbb{R}_+ \times H_w) \subset A_{H_w^0}^{\Psi_i}$ for $i = 1,\ldots,n$. Assume there are $M,L \geq 0$ such that for all $i = 1,\ldots,n$ and $t \in \mathbb{R}_+$ we have
\begin{align*}
\| \sigma_i(t,h) \|_w &\leq M, \quad h \in H_w
\\ \| \sigma_i(t,h_1) - \sigma_i(t,h_2) \|_w &\leq L \| h_1 - h_2 \|_w, \quad h_1,h_2 \in H_w. 
\end{align*}
Then, for each $h_0 \in H_w$, there exists a unique mild and a unique weak adapted c\`adl\`ag solution $(r_t)_{t \geq 0}$ to (\ref{hjmX}) with $r_0 = h_0$ satisfying
\begin{align}\label{sup-L2-finite-Hw}
\mathbb{E} \bigg[ \sup_{t \in [0,T]} \| r_t \|_w^2 \bigg] < \infty \quad \text{for all $T > 0$.}
\end{align}
\end{theorem}

\begin{proof}
By Proposition \ref{prop-const-C5}, $\alpha_{\rm HJM}$ maps into $H_w^0$, see (\ref{alpha-HJM-Hw}).
Since $\sigma = (\sigma_1,\ldots,\sigma_n) : \mathbb{R}_+ \times H_w \rightarrow \Pi_{i=1}^n A_{H_w^0}^{\Psi_i}$ is continuous by assumption and $\Sigma : \Pi_{i=1}^n A_{H_w^0}^{\Psi_i} \rightarrow H_w^0$ is continuous by Proposition \ref{prop-const-C5}, it follows that $\alpha_{\rm HJM} = \Sigma \circ \sigma$ is continuous. Moreover, by estimate (\ref{locally-Lipschitz}), we obtain for all $t \in \mathbb{R}_+$ and $h_1,h_2 \in H_w$ the estimation
\begin{align*}
\| \alpha_{\rm HJM}(t,h_1) - \alpha_{\rm HJM}(t,h_2) \|_w &\leq C_5 (1+M)^2 \sum_{i=1}^n \| \sigma_i(t,h_1) - \sigma_i(t,h_2) \|_w
\\ &\leq C_5 ( 1 + M )^2 n L \| h_1 - h_2 \|_w.
\end{align*}
Taking also into account Proposition \ref{prop-pseudo-contractive}, applying Theorem \ref{thm-sde-solution} finishes the proof.
\end{proof}

As an immediate consequence, we get the existence of L\'evy term structure models with constant direction volatilities.

\begin{corollary}\label{cor-main-on-hw-spaces}
Let $\sigma_i : \mathbb{R}_+ \times H_w \rightarrow H_w^0$ be defined by $\sigma_i(t,r) = \sigma_i(r) = \varphi_i(r) \lambda_i$, where $\lambda_i \in A_{H_w^0}^{\Psi_i}$ and $\varphi_i : H_w \rightarrow [0,1]$ for $i = 1,\ldots,n$.
Assume there is $L \geq 0$ such that for all $i = 1,\ldots,n$ we have
\begin{align*}
| \varphi_i(h_1) - \varphi_i(h_2) | &\leq L \| h_1 - h_2 \|_w, \quad h_1,h_2 \in H_w. 
\end{align*}
Then, for each $h_0 \in H_w$, there exists a unique mild and a unique weak adapted c\`adl\`ag solution $(r_t)_{t \geq 0}$ to (\ref{hjmX}) with $r_0 = h_0$ satisfying (\ref{sup-L2-finite-Hw}).
\end{corollary}

\begin{proof}
For all $h_1, h_2 \in H_w$ and all $i = 1,\ldots,n$ we get
\begin{align*}
\| \sigma_i(h_1) - \sigma_i(h_2) \|_w \leq L \| \lambda_i \|_w \| h_1 - h_2 \|_w.
\end{align*}
Also observing that $\| \sigma_i(h) \|_w \leq \| \lambda_i \|_w$ for all $h \in H_w$ and $i = 1,\ldots,n$, the proof is a straightforward consequence of Theorem \ref{thm-main-on-hw-spaces}.
\end{proof}

The only assumption on the driving L\'evy processes $X^1,\ldots,X^n$, in order to apply the previous results, is the exponential moments condition (\ref{canonical-cond0}). It is clearly satisfied for Brownian motions and Poisson processes.

There are also several purely discontinuous L\'evy processes fulfilling (\ref{canonical-cond0}), for instance generalized hyperbolic processes, which have been introduced by Barndorff-Nielsen \cite{Barndorff-Nielsen}, and their subclasses, namely the normal inverse Gaussian and hyperbolic processes. They have been applied to finance by Eberlein and co-authors in a series of papers, e.g. in \cite{Eberlein-Keller}.

Other purely discontinuous L\'evy processes satisfying (\ref{canonical-cond0}) are the generalized tempered stable processes, see \cite[Sec. 4.5]{Cont-Tankov}, which include Variance Gamma processes \cite{Madan}, CGMY processes \cite{CGMY} and bilateral Gamma processes \cite{Kuechler-Tappe}.

Consequently, Theorem \ref{thm-main-on-hw-spaces} applies to term structure models driven by any of the above types of L\'evy processes.

\section{Conclusion}\label{sec-conclusion}

We have established the existence of L\'evy term structure models on two spaces of forward curves, namely in Section \ref{sec-strong} on the Bj\"ork--Svensson space $H_{\beta,\gamma}$, on which $\frac{d}{dx}$ is a bounded linear operator, and in Section \ref{sec-mild} on the larger space $H_w$, where $\frac{d}{dx}$ becomes unbounded.

In Section \ref{sec-strong} it turned out that $H_{\beta,\gamma}$ is too small to assert that $\alpha_{\rm HJM}$ given by the HJM drift-condition (\ref{hjm-drift}) lies in $H_{\beta,\gamma}$. However, for certain jump-diffusion models we have established existence and uniqueness on this space, see Proposition \ref{prop-ex-1} and Proposition \ref{prop-ex-2}.

Our main results of Section \ref{sec-mild} (Theorem \ref{thm-main-on-hw-spaces} and Corollary \ref{cor-main-on-hw-spaces}), where we work on the larger space $H_w$, are applicable for a large range of driving L\'evy processes, including mixtures of Brownian motion and Poisson processes, and purely discontinuous L\'evy processes such as generalized hyperbolic processes and generalized tempered stable processes as well as several subclasses.

The existence results for L\'evy term structure models are based on a general result for Hilbert space valued stochastic equations, see Theorem \ref{thm-sde-solution} from the appendix. This result relies on two works of van Gaans \cite{Onno,Onno-Levy}. In order to make \cite[Thm. 4.1]{Onno-Levy} applicable for financial applications, where one is in particular interested in a solution with c\`adl\`ag trajectories, we have shown in the appendix that the stochastic integral constructed in van Gaans \cite{Onno-Levy} has a c\`adl\`ag modification and we have analyzed when it coincides with the usual It\^o-integral.

\begin{appendix}

\section{Overview and notation}\label{app-overview}

The goal of Appendix \ref{app-overview} -- Appendix \ref{app-SDE} is to provide an existence result for solutions of infinite dimensional stochastic differential equations, which is required in order to establish the existence of L\'evy term structure models.

We intend to apply a result of van Gaans \cite[Thm. 4.1]{Onno-Levy}. However, as we shall see in Section \ref{app-integration}, the stochastic integral $\text{(G-)}\int_0^t \Phi_s dX_s$ defined in van Gaans \cite{Onno-Levy} is not consistent with the usual It\^o-integral $\int_0^t \Phi_s dX_s$, which is used for financial modelling. This matters in view of applications to finance, because, as we have argued at the end of Section \ref{sec-intro}, we are in particular interested in a solution process with c\`adl\`ag paths.

In order to make \cite[Thm. 4.1]{Onno-Levy} applicable, we review the stochastic integral, which is defined in van Gaans \cite{Onno-Levy}, in Appendix \ref{app-integration}, show that it always possesses a c\`adl\`ag modification and analyze when it coincides with the usual It\^o-integral. In Appendix \ref{app-SDE}, we obtain the desired existence result concerning mild solutions, Theorem \ref{thm-sde-solution}, by applying \cite[Thm. 4.1]{Onno-Levy}. Using our findings of Appendix \ref{app-integration}, we additionally show that the solution has a c\`adl\`ag modification and that it is also a weak solution.

Let $H$ denote a separable Hilbert space with inner product $\langle \cdot,\cdot \rangle_H$ and associated norm $\| \cdot \|_H$. If there is no ambiguity, we shall simply write $\langle \cdot,\cdot \rangle$ and $\| \cdot \|$.

Let $T>0$ be a finite time horizon. We denote by $C_{\rm ad}([0,T];L^2(\Omega;H))$ the space of all continuous mappings $\Phi : [0,T] \rightarrow L^2(\Omega;H)$ which are also adapted.

For two stochastic processes $(\Phi_t)_{t \in [0,T]}$ and $(\Psi)_{t \in [0,T]}$ we say that $\Psi$ is a \textit{modification} of $\Phi$ if $\mathbb{P}(\Phi_t = \Psi_t) = 1$ for all $t \in [0,T]$.

An adapted $H$-valued process $(\Phi_t)_{t \in [0,T]}$ is called a \textit{martingale} if
\begin{itemize}
\item $\mathbb{E} \left[ \| \Phi_t \| \right] < \infty$ for all $t \in [0,T]$;

\item $\mathbb{E}[\Phi_t \, | \, \mathcal{F}_s] = \Phi_s$ ($\mathbb{P}$ -- a.s.) for all $0 \leq s \leq t \leq T$.
 
\end{itemize}
For the notion of conditional expectation of random variables having values in a separable Banach space, we refer to \cite[Sec. 1.3]{Da_Prato}.

An indispensable tool will be \textit{Doob's martingale inequality}
\begin{align}\label{doob-inequality}
\mathbb{E} \bigg[ \sup_{t \in [0,T]} \| \Phi_t \|^2 \bigg] \leq 4 \sup_{t \in [0,T]} \mathbb{E} \left[ \| \Phi_t \|^2 \right] = 4 \mathbb{E} \left[ \| \Phi_T \|^2 \right],
\end{align}
valid for every $H$-valued c\`adl\`ag martingale $\Phi$, which is a consequence of Thm. 3.8 and Prop. 3.7 in \cite{Da_Prato}.

\section{Stochastic integration}\label{app-integration}


Let $M$ be a real-valued L\'evy martingale satisfying $\mathbb{E}[M_1^2] < \infty$. We recall how in this case the stochastic integral $\text{(G-)} \int_0^t \Phi_s dM_s$, in the sense of van Gaans \cite[Sec. 3]{Onno-Levy}, is defined for $\Phi \in C_{\rm ad}([0,T];L^2(\Omega;H))$.

\begin{lemma}\label{lemma-approx-dX-int}
Let $\Phi \in C_{\rm ad}([0,T];L^2(\Omega;H))$. For each $t \in [0,T]$, there exists a unique random variable $Y_t \in L^2(\Omega;H)$ such that for every $\varepsilon > 0$ there exists $\delta > 0$ such that
\begin{align}\label{def-van-Gaans-int}
\mathbb{E} \left[ \bigg\| Y_t - \sum_{i=0}^{n-1} ( M_{t_{i+1}} - M_{t_i} ) \Phi_{t_i} \bigg\|^2 \right] < \varepsilon
\end{align}
for every partition $0 = t_0 < t_1 < \ldots < t_n = t$ with $\sup_{i=0,\ldots,n-1} |t_{i+1} - t_i| < \delta$.
\end{lemma}

\begin{proof}
The assertion is a consequence of \cite[Prop. 3.2.1]{Onno-Levy}.
\end{proof}

\begin{definition}\label{def-onno-int}
Let $\Phi \in C_{\rm ad}([0,T];L^2(\Omega;H))$. Then the stochastic integral $Y_t = \text{{\rm (G-)}}\int_0^t \Phi_s dM_s$, $t \in [0,T]$ is the stochastic process $Y = (Y_t)_{t \in [0,T]}$ where every $Y_t$ is the unique element from $L^2(\Omega;H)$ such that (\ref{def-van-Gaans-int}) is valid.
\end{definition}

We observe that for every $t \in [0,T]$ the stochastic integral ${\rm \text{(G-)}}\int_0^t \Phi_s dM_s$ is only determined up to a $\mathbb{P}$-null set. With regard to our applications to finance it arises the question if we can find a modification of the stochastic integral with c\`adl\`ag paths, a question which is not treated in \cite{Onno-Levy}.

Let $\Phi \in C_{\rm ad}([0,T];L^2(\Omega;H))$. We define 
\begin{align}\label{def-I}
I_t(\Phi) := \text{(G-)} \int_0^t \Phi_s dM_s, \quad t \in [0,T]
\end{align}
and the sequence of c\`adl\`ag adapted processes
\begin{align}\label{def-I-n}
I^n(\Phi) := \sum_{i=0}^{2^n - 1} (M^{t_{i+1}^n} - M^{t_i^n}) \Phi_{t_i^n}, \quad n \in \mathbb{N}_0
\end{align}
where we set for $n \in \mathbb{N}_0$ and $i \in \{ 0,\ldots,2^n \}$
\begin{align}\label{def-t-i-n}
t_i^n := i 2^{-n} T,
\end{align} 
that is, we have a sequence of dyadic decompositions of the interval $[0,T]$.
Note that each $I^n(\Phi)$ is a martingale and that for each $t \in [0,T]$ we have $I_t^n(\Phi) \rightarrow I_t(\Phi)$ in $L^2(\Omega;H)$ by Lemma \ref{lemma-approx-dX-int}.

We let $\mathcal{M}^2$ be the linear space of all c\`adl\`ag $H$-valued martingales $(\Phi_t)_{t \in [0,T]}$, which are square-integrable, i.e. $\mathbb{E} \left[ \| \Phi_t \|^2 \right] < \infty$ for all $t \in [0,T]$, equipped with the norm
\begin{align*}
\| \Phi \|_2 = \mathbb{E} \bigg[ \sup_{t \in [0,T]} \| \Phi_t \|^2 \bigg]^{\frac{1}{2}}.
\end{align*}
Note that by Doob's martingale inequality (\ref{doob-inequality}), $\| \Phi \|_2$ is finite for every $\Phi \in \mathcal{M}^2$, and therefore $\| \cdot \|_2$ defines a norm on the linear space $\mathcal{M}^2$. For the next result, we can almost literally follow the proof of \cite[Prop. 3.9]{Da_Prato}, which considers the continuous time case. For convenience of the reader, we provide the proof here.

\begin{proposition}\label{prop-m2-banach}
The normed space $(\mathcal{M}^2,\| \cdot \|_2)$ is a Banach space.
\end{proposition}

\begin{proof}
Let $(\Phi^n)$ be a Cauchy sequence in $\mathcal{M}^2$, i.e. for every $\varepsilon > 0$ there is an index $n_0 \in \mathbb{N}$ such that
\begin{align}\label{Cauchy-in-M2}
\mathbb{E} \bigg[ \sup_{t \in [0,T]} \| \Phi_t^n - \Phi_t^m \|^2 \bigg] < \varepsilon \quad \text{for all $n,m \geq n_0$.}
\end{align}
By the Markov inequality, there exists a subsequence $(\Phi^{n_k})$ such that
\begin{align*}
\mathbb{P} \bigg( \sup_{t \in [0,T]} \| \Phi_t^{n_{k+1}} - \Phi_t^{n_k} \| \geq 2^{-k} \bigg) \leq 2^{-k} \quad \text{for all $k \in \mathbb{N}$.}
\end{align*}
The Borel-Cantelli lemma implies that for almost all $\omega \in \Omega$ the sequence $(\Phi^{n_k}(\omega))$ is a Cauchy sequence in the space of c\`adl\`ag functions on $[0,T]$ equipped with the supremum-norm. Therefore, $(\Phi^{n_k})$ converges $\mathbb{P}$--a.s. to an adapted process $\Phi$, uniformly on $[0,T]$. Hence, $\Phi$ is c\`adl\`ag. 

For each $t \in [0,T]$, the convergence $\Phi_t^{n_k} \rightarrow \Phi_t$ is valid in $L^2(\Omega;H)$, because $(\Phi_t^n)$ is a Cauchy sequence in $L^2(\Omega;H)$ by (\ref{Cauchy-in-M2}). For $0 \leq s \leq t \leq T$ and $k \in \mathbb{N}$ we have
$\mathbb{E}[\Phi_t^{n_k} \,|\, \mathcal{F}_s] = \Phi_s^{n_k}$ ($\mathbb{P}$--a.s.), implying that $\mathbb{E}[\Phi_t \,|\, \mathcal{F}_s] = \Phi_s$ ($\mathbb{P}$--a.s.). Consequently, $\Phi$ is a martingale, and by Doob's martingale inequality (\ref{doob-inequality}), we get
\begin{align*}
\mathbb{E} \bigg[ \sup_{t \in [0,T]} \| \Phi_t - \Phi_t^n \|^2 \bigg] \leq 4 \sup_{t \in [0,T]} \mathbb{E} \left[ \| \Phi_t - \Phi_t^n \|^2 \right] = 4 \mathbb{E} \left[ \| \Phi_T - \Phi_T^n \|^2 \right] \rightarrow 0
\end{align*}
by (\ref{Cauchy-in-M2}) and completeness of $L^2(\Omega;H)$, i.e. $\Phi^n \rightarrow \Phi$ in $\mathcal{M}^2$.
\end{proof}

In the following auxiliary result, $\langle M,M \rangle$ denotes the \textit{predictable quadratic covariation} of the real-valued square-integrable martingale $M$, see \cite[Thm. I.4.2]{JS}.

\begin{lemma}\label{lemma-covariation}
Let $0 = t_0 < \ldots < t_n = T$ and $Z_i : \Omega \rightarrow H$ be $\mathcal{F}_{t_i}$-measurable for $i = 0,\ldots,n-1$. Then we have
\begin{align*}
\mathbb{E} \left[ \bigg\| \sum_{i=0}^{n-1} (M_{t_{i+1}} - M_{t_i}) Z_i \bigg\|^2 \right] = \mathbb{E} \left[ \sum_{i=0}^{n-1} (\langle M,M \rangle_{t_{i+1}} - \langle M,M \rangle_{t_i}) \| Z_i \|^2 \right].
\end{align*}
\end{lemma}

\begin{proof}
By using the identity $\| x \|^2 = \langle x,x \rangle_H$, $x \in H$ we obtain that
\begin{align*}
&\bigg\| \sum_{i=0}^{n-1} (M^{t_{i+1}} - M^{t_i}) Z_i \bigg\|^2 - \sum_{i=0}^{n-1} (\langle M,M \rangle^{t_{i+1}} - \langle M,M \rangle^{t_i}) \| Z_i \|^2
\\ &= 2 \sum_{i,j=0 \atop i < j}^{n-1} (M^{t_{i+1}} - M^{t_i}) (M^{t_{j+1}} - M^{t_j}) \langle Z_i, Z_j \rangle_H
\\ &\quad + \sum_{i=0}^{n-1} \left[ ( M^{t_{i+1}} )^2 - \langle M,M \rangle^{t_{i+1}} - ( M^{t_i} )^2 + \langle M,M \rangle^{t_i} - 2 M_{t_i} (M^{t_{i+1}} - M^{t_i}) \right] \| Z_i \|^2
\end{align*}
is a martingale. Since $\sum_{i=0}^{n-1} (\langle M,M \rangle^{t_{i+1}} - \langle M,M \rangle^{t_i}) \| Z_i \|^2$ is continuous and therefore predictable, the uniqueness of the predictable quadratic covariation yields
\begin{align*}
\Big\langle \bigg\| \sum_{i=0}^{n-1} (M^{t_{i+1}} - M^{t_i}) Z_i \bigg\|, \bigg\| \sum_{i=0}^{n-1} (M^{t_{i+1}} - M^{t_i}) Z_i \bigg\| \Big\rangle = \sum_{i=0}^{n-1} (\langle M,M \rangle^{t_{i+1}} - \langle M,M \rangle^{t_i}) \| Z_i \|^2,
\end{align*}
proving the claimed equation.
\end{proof}

\begin{theorem}\label{thm-cadlag-mod}
Let $\Phi \in C_{\rm ad}([0,T];L^2(\Omega;H))$. Then $I(\Phi)$ has a modification which belongs to $\mathcal{M}^2$ and, moreover, $I^n(\Phi) \rightarrow I(\Phi)$ in $\mathcal{M}^2$.
\end{theorem}

\begin{proof}
Let $\varepsilon > 0$ be arbitrary. Since $\Phi : [0,T] \rightarrow L^2(\Omega;H)$ is uniformly continuous on the compact interval $[0,T]$, there exists $\delta > 0$ such that
\begin{align}\label{est-c-F}
\mathbb{E} \left[ \| \Phi_t - \Phi_s \|^2 \right] < \frac{\varepsilon}{4T \left( c + \int_{\mathbb{R}} x^2 F(dx) \right)}
\end{align}
for all $s,t \in [0,T]$ with $|t-s| < \delta$, where $c$ denotes the Gaussian part and $F$ the L\'evy measure of $M$. Choose $n_0 \in \mathbb{N}$ such that $2^{-n_0} T < \delta$.
For all $n,m \in \mathbb{N}_0$ with $n > m \geq n_0$ we obtain
\begin{align*}
I^n(\Phi) - I^m(\Phi) = \sum_{i=0}^{2^n - 1} \big( M^{t_{i+1}^n} - M^{t_i^n} \big) \big( \Phi_{t_i^n} - \Phi_{t_{j(i)}^n} \big) 
\end{align*}
with $j(i) \in \{ 0,\ldots,i \}$ such that $|t_i^n - t_{j(i)}^n| < 2^{-n_0} T < \delta$ for all $i = 0,\ldots,2^n - 1$.
We obtain by Doob's martingale inequality (\ref{doob-inequality}), Lemma \ref{lemma-covariation} and (\ref{est-c-F}) for all $n,m \in \mathbb{N}_0$ with $n > m \geq n_0$
\begin{align*}
&\mathbb{E} \bigg[ \sup_{t \in [0,T]} \| I_t^n(\Phi) - I_t^m(\Phi) \|^2 \bigg] \leq 4 \mathbb{E} \left[ \bigg\| \sum_{i=0}^{2^n - 1} \big( M_{t_{i+1}^n} - M_{t_i^n} \big) \big( \Phi_{t_i^n} - \Phi_{t_{j(i)}^n} \big) \bigg\|^2 \right]
\\ &= 4 \sum_{i=0}^{2^n - 1} \mathbb{E} \left[ \big( \langle M,M \rangle_{t_{i+1}^n} - \langle M,M \rangle_{t_i^n} \big)  \| \Phi_{t_i^n} - \Phi_{t_{j(i)}^n} \|^2 \right]
\\ &= 4 \left( c + \int_{\mathbb{R}} x^2 F(dx) \right) \sum_{i=0}^{2^n - 1} (t_{i+1}^n - t_i^n) \mathbb{E} \left[ \| \Phi_{t_i^n} - \Phi_{t_{j(i)}^n} \|^2 \right] < \varepsilon.
\end{align*}
The latter identity is valid, because $\langle M,M \rangle$ is the compensator of $[M,M]$ by \cite[Prop. I.4.50.b]{JS} and because the relation $[M,M]_t = ct + \sum_{s \leq t} \Delta M_s^2$ is valid according to \cite[Thm. I.4.52]{JS}.

Thus, the sequence $(I^n(\Phi))$ is a Cauchy sequence in $\mathcal{M}^2$. Proposition \ref{prop-m2-banach} and Lemma \ref{lemma-approx-dX-int} complete the proof.
\end{proof}

For $\Phi \in C_{\rm ad}([0,T];L^2(\Omega;H))$, the integral with respect to $dt$ can, according to \cite[Lemma 3.6]{Onno-Levy}, be defined as a Riemann integral. More precisely:

\begin{lemma}\label{lemma-approx-dt}
Let $\Phi \in C_{\rm ad}([0,T];L^2(\Omega;H))$. For each $t \in [0,T]$, there exists a unique random variable $Y_t \in L^2(\Omega;H)$ such that for every $\varepsilon > 0$ there exists $\delta > 0$ such that
\begin{align}\label{def-van-Gaans-int-dt}
\mathbb{E} \left[ \bigg\| Y_t - \sum_{i=0}^{n-1} ( t_{i+1} - t_i ) \Phi_{t_i} \bigg\|^2 \right] < \varepsilon
\end{align}
for every partition $0 = t_0 < t_1 < \ldots < t_n = t$ with $\sup_{i=0,\ldots,n-1} |t_{i+1} - t_i| < \delta$.
\end{lemma}

\begin{proof}
Fix $t \in [0,T]$ and let $\varepsilon > 0$ be arbitrary. Since $\Phi : [0,t] \rightarrow L^2(\Omega;H)$ is uniformly continuous on the compact interval $[0,t]$, there exists $\delta > 0$ such that 
\begin{align}\label{delta-choice-Riemann}
\mathbb{E} \left[ \| \Phi_s - \Phi_r \|^2 \right] < \frac{\varepsilon}{t^2}
\end{align}
for all $r,s \in [0,t]$ with $|s-r| < \delta$.

Let $Z_1 = \{ 0 = t_0 < t_1 < \ldots < t_n = t \}$ and $Z_2 = \{ 0 = s_0 < s_1 < \ldots < s_m = t \}$ be two decompositions satisfying $\sup_{i=0,\ldots,n-1} |t_{i+1} - t_i| < \delta$ and $\sup_{i=0,\ldots,m-1} |s_{i+1} - s_i| < \delta$. Then there is a unique decomposition $Z = \{ 0 = r_0 < r_1 < \ldots < r_p = t \}$ such that $Z = Z_1 \cup Z_2$. Thus, we get
\begin{align*}
\sum_{i=0}^{n-1} (t_{i+1} - t_i) \Phi_{t_i} - \sum_{i=0}^{m-1} (s_{i+1} - s_i) \Phi_{s_i} = \sum_{i=0}^{p-1} (r_{i+1} - r_i) (\Phi_{a_i} - \Phi_{b_i})
\end{align*}
with $a_i \in Z_1$, $b_i \in Z_2$ and $| b_i - a_i| < \delta$ for all $i = 0,\ldots,p-1$. We obtain by the Cauchy-Schwarz inequality and (\ref{delta-choice-Riemann})
\begin{align*}
&\mathbb{E} \left[ \bigg\| \sum_{i=0}^{n-1} (t_{i+1} - t_i) \Phi_{t_i} - \sum_{i=0}^{m-1} (s_{i+1} - s_i) \Phi_{s_i} \bigg\|^2 \right] 
\leq \mathbb{E} \bigg[ \bigg( \sum_{i=0}^{p-1} (r_{i+1} - r_i) \| (\Phi_{a_i} - \Phi_{b_i}) \| \bigg)^2 \bigg]
\\ &\leq \mathbb{E} \bigg[ \bigg( \sum_{i=0}^{p-1} (r_{i+1} - r_i) \bigg) \bigg( \sum_{i=0}^{p-1} (r_{i+1} - r_i) \| \Phi_{a_i} - \Phi_{b_i} \|^2 \bigg) \bigg]
\\ &= t \sum_{i=0}^{p-1} (r_{i+1} - r_i) \mathbb{E} \left[ \| \Phi_{a_i} - \Phi_{b_i} \|^2 \right] < \varepsilon.
\end{align*}
By the completeness of $L^2(\Omega;H)$, the lemma is proven.
\end{proof}

\begin{definition}
Let $\Phi \in C_{\rm ad}([0,T];L^2(\Omega;H))$. Then the integral $Y_t = \text{{\rm (G-)}}\int_0^t \Phi_s ds$, $t \in [0,T]$ is the stochastic process $Y = (Y_t)_{t \in [0,T]}$ where every $Y_t$ is the unique element from $L^2(\Omega;H)$ such that (\ref{def-van-Gaans-int-dt}) is valid.
\end{definition}

Again, for every $t \in [0,T]$ the integral $\text{{\rm (G-)}}\int_0^t \Phi_s ds$ is only determined up to a $\mathbb{P}$--null set. We shall prove the existence of a continuous modification.

Let $\Phi \in C_{\rm ad}([0,T],L^2(\Omega;H))$. We define 
\begin{align}\label{def-J}
J_t(\Phi) := \text{(G-)} \int_0^t \Phi_s ds, \quad t \in [0,T]
\end{align}
and the sequence of continuous adapted processes
\begin{align}\label{def-J-n}
J_t^n(\Phi) := \sum_{i=0}^{2^n - 1} (t_{i+1}^n \wedge t - t_i^n \wedge t) \Phi_{t_i^n}, \quad n \in \mathbb{N}_0
\end{align}
where the $t_i^n$ are defined in (\ref{def-t-i-n}). Note that for each $t \in [0,T]$ we have $J_t^n(\Phi) \rightarrow J_t(\Phi)$ in $L^2(\Omega;H)$ by Lemma \ref{lemma-approx-dt}.

\begin{theorem}\label{thm-cont-mod}
Let $\Phi \in C_{\rm ad}([0,T];L^2(\Omega;H))$. Then $J(\Phi)$ has a continuous modification and, moreover,
\begin{align*}
\sup_{t \in [0,T]} \| J_t^n(\Phi) - J_t(\Phi) \| \rightarrow 0 \quad \text{$\mathbb{P}$--a.s.}
\end{align*}
\end{theorem}

\begin{proof}
Let $\varepsilon > 0$ be arbitrary. Since $\Phi : [0,T] \rightarrow L^2(\Omega;H)$ is uniformly continuous on the compact interval $[0,T]$, there exists $\delta > 0$ such that
\begin{align}\label{est-T-square}
\mathbb{E} \left[ \| \Phi_t - \Phi_s \|^2 \right] < \frac{\varepsilon}{T^2}
\end{align}
for all $s,t \in [0,T]$ with $|t-s| < \delta$. Choose $n_0 \in \mathbb{N}$ such that $2^{-n_0} T < \delta$.
For all $n,m \in \mathbb{N}_0$ with $n > m \geq n_0$ we obtain
\begin{align*}
J^n(\Phi) - J^m(\Phi) = \sum_{i=0}^{2^n - 1} \big( t_{i+1}^n \wedge t - t_i^n \wedge t \big) \big( \Phi_{t_i^n} - \Phi_{t_{j(i)}^n} \big) 
\end{align*}
with $j(i) \in \{ 0,\ldots,i \}$ such that $|t_i^n - t_{j(i)}^n| < 2^{-n_0} T < \delta$ for all $i = 0,\ldots,2^n - 1$.
We obtain by the Cauchy-Schwarz inequality and (\ref{est-T-square}) for all $n,m \geq n_0$ with $n > m$
\begin{align*}
\mathbb{E} \bigg[ \sup_{t \in [0,T]} \| J_t^n(\Phi) - J_t^m(\Phi) \|^2 \bigg] &= \mathbb{E} \left[ \sup_{t \in [0,T]} \bigg\| \sum_{i=0}^{2^n - 1} \big( t \wedge t_{i+1}^n - t \wedge t_i^n \big) \big( \Phi_{t_i^n} - \Phi_{t_{j(i)}^n} \big) \bigg\|^2 \right]
\\ &\leq T \mathbb{E} \bigg[ \sup_{t \in [0,T]} \sum_{i=0}^{2^n - 1} \big( t \wedge t_{i+1}^n - t \wedge t_i^n \big) \| \Phi_{t_i^n} - \Phi_{t_{j(i)}^n} \|^2 \bigg]
\\ &= T \sum_{i=0}^{2^n - 1} \big( t_{i+1} - t_i \big) \mathbb{E} \left[ \| \Phi_{t_i^n} - \Phi_{t_{j(i)}^n} \|^2 \right] < \varepsilon.
\end{align*}
By the Markov inequality, there exists a subsequence $(J^{n_k}(\Phi))$ such that
\begin{align*}
\mathbb{P} \bigg( \sup_{t \in [0,T]} \| J_t^{n_{k+1}}(\Phi) - J_t^{n_k}(\Phi) \| \geq 2^{-k} \bigg) \leq 2^{-k} \quad \text{for all $k \in \mathbb{N}$.}
\end{align*}
The Borel-Cantelli lemma implies that for almost all $\omega \in \Omega$ the sequence $(J^{n_k}(\Phi)(\omega))$ is a Cauchy sequence in the space of continuous functions on $[0,T]$ equipped with the supremum-norm. Therefore, $(J^{n_k}(\Phi))$ converges $\mathbb{P}$--a.s. to an adapted process, uniformly on $[0,T]$, which is therefore continuous.

According to Lemma \ref{lemma-approx-dt}, this limit process is a modification of the integral process $J(\Phi)$.
\end{proof}

Now let $X$ be a real-valued L\'evy process with $\mathbb{E}[X_1^2] < \infty$. Then it admits a unique decomposition $X_t = M_t + bt$, where $M$ is a L\'evy martingale satisfying $\mathbb{E}[M_1^2] < \infty$ and $b = \mathbb{E}[X_1]$. According to \cite[Def. 3.7]{Onno-Levy}, we set
\begin{align*}
\text{(G-)}\int_0^t \Phi_s dX_s := \text{(G-)}\int_0^t \Phi_s dM_s + b \cdot (\text{G-}) \int_0^t \Phi_s ds.
\end{align*}
We shall also use the notation
\begin{align*}
G(\Phi)_t = \text{(G-)} \int_0^t \Phi_s dX_s, \quad t \in [0,T].
\end{align*}
Note that $G(\Phi) = I(\Phi) + b \cdot J(\Phi)$, where $I(\Phi)$ is defined in (\ref{def-I}) and $J(\Phi)$ is defined in (\ref{def-J}). We also introduce $G^n(\Phi) = I^n(\Phi) + b \cdot J^n(\Phi)$ for $n \in \mathbb{N}_0$, where $I^n(\Phi)$ is defined in (\ref{def-I-n}) and $J^n(\Phi)$ is defined in (\ref{def-J-n}).

For a predictable $H$-valued process $\Phi$ and a real-valued semimartingale $X$, we can define the usual It\^o-integral (developed e.g. in Jacod and Shiryaev \cite{JS} or Protter \cite{Protter})
\begin{align*}
\int_0^t \Phi_s dX_s,
\end{align*}
which is used for financial modelling. The construction is just as for real-valued integrands, namely by defining the integral first for simple integrands and then extending it via the It\^o-isometry. In order to get the It\^o-isometry, it is vital that the state space $H$ is a Hilbert space. 

The construction of the stochastic integral in the more general situation, where the driving semimartingale may also be infinite dimensional, can be found in M\'etivier \cite{Metivier}. Da~Prato and Zabczyk \cite{Da_Prato} and Carmona and Tehranchi \cite{carteh} treat the case with infinite dimensional Brownian motion as integrator, in \cite{carteh} also with a focus on interest rate models.

We also remark that the stochastic integral can still be defined on appropriate Banach spaces, so-called M-type 2 spaces. Then the integral is still a bounded linear operator, but no isometry, in general. We refer to \cite{barbara-integration} for further details.



We now observe that the integral $\text{(G-)}\int_0^t \Phi_s dX_s$ of van Gaans \cite{Onno-Levy} is not consistent with the usual stochastic integral $\int_0^t \Phi_s dX_s$ used in financial modelling. As an example, let $X$ be a standard Poisson process with values in $\mathbb{R}$. In Ex. 3.9 in \cite{Onno-Levy} it is derived that
\begin{align*}
{\rm \text{(G-)}}\int_0^t X_s dX_s = \frac{1}{2} \left( X_t^2 - X_t \right).
\end{align*}
Apparently, this does not coincide with the pathwise Lebesgue-Stieltjes integral
\begin{align*}
\int_0^t X_s dX_s = \frac{1}{2} \left( X_t^2 + X_t \right),
\end{align*} 
but we have
\begin{align*}
{\rm \text{(G-)}}\int_0^t X_s dX_s = \int_0^t X_{s-} dX_s,
\end{align*}
showing that inconsistencies occur as soon as integrands with jumps are used. Indeed, we have the following general result about the relation between the integral of van Gaans and the usual It\^o-integral:


\begin{theorem}\label{thm-vanGaans-versus-usual}
Let $\Phi \in C_{\rm ad}([0,T];L^2(\Omega;H))$ be left-continuous or c\`adl\`ag. Then we have for all $t \in [0,T]$
\begin{align}\label{coincidence-2}
\text{{\em (G-)}}\int_0^t \Phi_s dX_s = \int_0^t \Phi_{s-} dX_s \quad \text{$\mathbb{P}$--a.s.}
\end{align}
\end{theorem}

\begin{proof}
If $\Phi$ is left-continuous, we have $\sup_{t \in [0,T]} \| G_t^n - G_t \| \rightarrow 0$ almost surely by Theorem \ref{thm-cadlag-mod} and Theorem \ref{thm-cont-mod}, and therefore also in probability. For the usual It\^o-integral we have
$\sup_{t \in [0,T]} \left\| G_t^n - \int_0^t \Phi_s dX_s \right\| \rightarrow 0$
in probability, which is proven as in the real-valued case, see e.g. \cite[Prop. I.4.44]{JS}. Thus we obtain for all $t \in [0,T]$
\begin{align}\label{coincidence-1}
{\rm \text{(G-)}}\int_0^t \Phi_s dX_s = \int_0^t \Phi_s dX_s \quad \text{$\mathbb{P}$--a.s.}
\end{align} 
and, since $\Phi$ is left-continuous, also relation (\ref{coincidence-2}).

If $\Phi$ is c\`adl\`ag, we show that $\Phi_-$ is a modification of $\Phi$, because then (\ref{coincidence-2}) is a consequence of (\ref{coincidence-1}) and Lemma \ref{lemma-modification} below. Let $t \in (0,T]$ be arbitrary and $(t_n)$ be a sequence such that $t_n \uparrow t$. Since $\Phi : [0,T] \rightarrow L^2(\Omega;H)$ is continuous, we deduce $\mathbb{E}[ \| \Phi_t - \Phi_{t_n} \|^2 ] \rightarrow 0$. Thus there is a subsequence $(n_k)$ with $\| \Phi_t - \Phi_{t_{n_k}} \| \rightarrow 0$ almost surely, and therefore we have $\mathbb{P}(\Phi_t = \Phi_{t-}) = 1$.
\end{proof}

\begin{remark}
If the driving process $X$ is a (possibly infinite dimensional) Brownian motion, the equivalence of the van Gaans integral with the usual stochastic integral (see Da~Prato and Zabczyk \cite{Da_Prato} for the infinite dimensional case) is provided in \cite[Sec. 3]{Onno}.
\end{remark}

It remains to show the following auxiliary result, which we have used in the proof of Theorem \ref{thm-vanGaans-versus-usual}.

\begin{lemma}\label{lemma-modification}
Let $\Phi, \Psi \in C_{\rm ad}([0,T];L^2(\Omega;H))$ be such that $\Psi$ is a modification of $\Phi$. Then $G(\Psi)$ is a modification of $G(\Phi)$.
\end{lemma}

\begin{proof}
Let $t \in [0,T]$ be arbitrary. By hypothesis, we have $\mathbb{P}(G_t^n(\Phi) = G_t^n(\Psi)) = 1$ for all $n \in \mathbb{N}_0$. Since $G_t^n(\Phi) \rightarrow G_t(\Phi)$ and $G_t^n(\Psi) \rightarrow G_t(\Psi)$ in $L^2(\Omega;H)$ by Lemma \ref{lemma-approx-dX-int} and Lemma \ref{lemma-approx-dt}, there is a subsequence $(n_k)$ such that $G_t^{n_k}(\Phi) \rightarrow G_t(\Phi)$ almost surely, and another subsequence $n_{k_l}$ such that $G_t^{n_{k_l}}(\Psi) \rightarrow G_t(\Psi)$ almost surely, showing that $\mathbb{P}(G_t(\Phi) = G_t(\Psi)) = 1$.
\end{proof}

\section{Stochastic differential equations}\label{app-SDE}

Now let $(S_t)_{t \geq 0}$ be a $C_0$-semigroup in the separable Hilbert space $H$, i.e. a family of bounded linear operators $S_t : H \rightarrow H$ such that
\begin{itemize}
\item $S_0 = {\rm Id}$;

\item $S_{s+t} = S_s S_t$ for all $s,t \geq 0$;

\item $\lim_{t \rightarrow 0} S_t h = h$ for all $h \in H$;
\end{itemize}
with generator $A : \mathcal{D}(A) \subset H \rightarrow H$. By $\| \cdot \|_{\mathcal{L}(H)}$ we denote the operator norm of a bounded linear operator. The semigroup $(S_t)$ is called \textit{contractive in $H$} if
\begin{align*}
\| S_t \|_{\mathcal{L}(H)} \leq 1, \quad t \geq 0
\end{align*}
and \textit{pseudo-contractive in $H$} if there is a constant $\omega \geq 0$ such that
\begin{align*}
\| S_t \|_{\mathcal{L}(H)} \leq e^{\omega t}, \quad t \geq 0.
\end{align*}
In this section, we intend to find mild solutions of stochastic differential equations of the type
\begin{align}\label{SDE-to-solve}
\left\{
\begin{array}{rcl}
dr_t & = & (A r_t + \alpha(t,r_t)) dt + \sum_{i=1}^n \sigma_i(t,r_{t-}) dX_t^i, \medskip
\\ r_0 & = & h_0
\end{array}
\right.
\end{align}
driven by real-valued L\'evy processes $X^1,\ldots,X^n$ satisfying $\mathbb{E}[(X_1^i)^2] < \infty$, $i = 1,\ldots,n$, for each initial condition $h_0 \in H$, that is, a process $(r_t)_{t \geq 0}$ satisfying
\begin{align}\label{mild-solution}
r_t = S_t h_0 + \int_0^t S_{t-s} \alpha(s,r_s) ds + \sum_{i=1}^n \int_0^t S_{t-s} \sigma_i(s,r_{s-}) dX_s^i, \quad t \in \mathbb{R}_+.
\end{align}
We also intend to establish the existence of a weak solution $(r_t)_{t \geq 0}$ to (\ref{SDE-to-solve}), i.e. $(r_t)$ satisfies, for all $\zeta \in \mathcal{D}(A^*)$,
\begin{align}\label{weak-solution}
\langle \zeta, r_t \rangle = \langle \zeta, h_0 \rangle + \int_0^t \Big( \langle A^* \zeta, r_s \rangle + \langle \zeta, \alpha(s,r_s) \rangle \Big) ds + \sum_{i=1}^n \int_0^t \langle \zeta, \sigma_i(s,r_{s-}) \rangle dX_s^i
\end{align}
for each $t \in \mathbb{R}_+$. By convention, uniqueness of a solution to (\ref{SDE-to-solve}) is meant up to a modification. Here is our main existence and uniqueness result:

\begin{theorem}\label{thm-sde-solution}
Let $(S_t)_{t \geq 0}$ be a $C_0$-semigroup in $H$, and $H_0 \subset H$ be a closed subspace such that $(S_t)$ is pseudo-contractive in $H_0$. Let $\alpha,\sigma_1,\ldots,\sigma_n : \mathbb{R}_+ \times H \rightarrow H_0$ be continuous. Assume there is constant $L \geq 0$ such that
\begin{align}\label{lip-alpha-van-Gaans}
\| \alpha(t,h_1) - \alpha(t,h_2) \| &\leq L \| h_1 - h_2 \|
\\ \label{lip-sigma-van-Gaans} \| \sigma_i(t,h_1) - \sigma_i(t,h_2) \| &\leq L \| h_1 - h_2 \|, \quad i = 1,\ldots,n
\end{align}
for all $t \in \mathbb{R}_+$ and all $h_1, h_2 \in H$. Then, for each $h_0 \in H$, there exists a unique mild and a unique weak adapted c\`adl\`ag solution $(r_t)_{t \geq 0}$ to (\ref{SDE-to-solve}) with $r_0 = h_0$ satisfying
\begin{align}\label{solution-in-S2}
\mathbb{E} \bigg[ \sup_{t \in [0,T]} \| r_t \|^2 \bigg] < \infty \quad \text{for all $T > 0$.}
\end{align}
\end{theorem}

\begin{proof}
Let $h_0 \in H$ be arbitrary.
We decompose each L\'evy process $X_t^i = M_t^i + b_i t$ into its martingale and finite variation part, where we notice that $b_i = \mathbb{E}[X_1^i]$. By \cite[Thm. 4.1]{Onno-Levy} there exists a unique adapted continuous function $r : \mathbb{R}_+ \rightarrow L^2(\Omega;H)$ such that for all $t \geq 0$
\begin{align*}
r_t = S_t h_0 + \text{(G-)} \int_0^t S_{t-s} \tilde{\alpha}(s,r_s) ds + \sum_{i=1}^n \text{(G-)} \int_0^t S_{t-s} \sigma_i(s,r_s) dM_s^i,
\end{align*}
where $\tilde{\alpha}(t,r) = \alpha(t,r) + \sum_{i=1}^n b_i \sigma_i(t,r)$. By assumption, $(S_t)$ is pseudo-contractive in $H_0$. Hence there exists a constant $\omega \geq 0$ such that the $C_0$-semigroup $(T_t)_{t \geq 0}$ defined as
\begin{align}\label{def-new-semigroup} 
T_t := e^{-\omega t} S_t, \quad t \in \mathbb{R}_+ 
\end{align} 
is contractive in $H_0$. 
By the Szek\"ofalvi-Nagy's theorem on unitary dilations (see e.g. \cite[Thm. I.8.1]{Nagy}, or \cite[Sec. 7.2]{Davies}), there exists another separable Hilbert space $\mathcal{H}_0$ and a strongly continuous unitary group $(U_t)_{t \in \mathbb{R}}$ in $\mathcal{H}_0$ such that the diagram
\[ \begin{CD}
\mathcal{H}_0 @>U_t>> \mathcal{H}_0\\
@AA\ell A @VV\pi V\\
H_0 @>T_t>> H_0
\end{CD} \]
commutes for every $t \in \mathbb{R}_+$, where $\ell : H_0 \rightarrow \mathcal{H}_0$ is an isometric embedding (hence the adjoint operator $\pi := \ell^*$ is the orthogonal projection from $\mathcal{H}_0$ into $H_0$), that is 
\begin{align}\label{diagram-commutes}
\pi U_t \ell h = T_t h \quad \text{for all $t \in \mathbb{R}_+$ and $h \in H_0$.}
\end{align} 
Using (\ref{def-new-semigroup}), (\ref{diagram-commutes}) and \cite[Thm. 3.3.3]{Onno-Levy} we obtain for all $i = 1,\ldots,n$ and $t \geq 0$
\begin{align*}
&\text{(G-)}\int_0^t S_{t-s} \sigma_i(s,r_s) dM_s^i = \text{(G-)}\int_0^t e^{\omega(t-s)} T_{t-s} \sigma_i(s,r_s) dM_s^i 
\\ &= e^{\omega t} \text{(G-)}\int_0^t e^{-\omega s} \pi U_{t-s} \ell \sigma_i(s,r_s) dM_s^i
= e^{\omega t} \pi U_t \text{(G-)}\int_0^t e^{-\omega s} U_{-s} \ell \sigma_i(s,r_s) dM_s^i.
\end{align*}
The integral process 
\begin{align*}
\text{(G-)}\int_0^t e^{-\omega s} U_{-s} \ell \sigma_i(s,r_s) dM_s^i
\end{align*}
has a c\`adl\`ag modification by Theorem \ref{thm-cadlag-mod}. Thus the process
\begin{align*}
 \text{(G-)}\int_0^t S_{t-s} \sigma_i(s,r_s) dM_s^i
\end{align*}
has a c\`adl\`ag modification, because $(t,h) \mapsto U_t h$ is uniformly continuous on compact subsets, see e.g. \cite[Lemma I.5.2]{Engel-Nagel}.

A similar argumentation, using Theorem \ref{thm-cont-mod}, shows that
\begin{align*}
 \text{(G-)}\int_0^t S_{t-s} \tilde{\alpha}(s,r_s) ds
\end{align*}
has a continuous modification.

Therefore, $(r_t)$ has a c\`adl\`ag modification, and, by Theorem \ref{thm-vanGaans-versus-usual}, it satisfies
\begin{align*}
r_t = S_t h_0 + \int_0^t S_{t-s} \tilde{\alpha}(s,r_s) ds + \sum_{i=1}^n \int_0^t S_{t-s} \sigma_i(s,r_{s-}) dM_s^i, \quad t \geq 0.
\end{align*}
Consequently, $(r_t)_{t \geq 0}$ is a mild solution to (\ref{SDE-to-solve}), i.e. it satisfies (\ref{mild-solution}). Introducing the processes
\begin{align}
\label{def-Phi-weak} \Phi_t &:= \int_0^t S_{t-s} \tilde{\alpha}(s,r_s) ds,
\\ \label{def-Psi-weak} \Psi_t^i &:= \int_0^t S_{t-s} \sigma_i(s,r_{s-}) dM_s^i, \quad i = 1,\ldots,n
\end{align}
we have by our findings above
\begin{align}\label{short-solution}
r_t = S_t h_0 + \Phi_t + \sum_{i=1}^n \Psi_t^i, \quad t \geq 0.
\end{align}
We fix an arbitrary $T > 0$. By (\ref{def-new-semigroup}), (\ref{diagram-commutes}) and noting that $\| \pi \|_{\mathcal{L}(\mathcal{H}_0;H_0)} = 1$ and $\| U_t \|_{\mathcal{L}(\mathcal{H}_0)} \leq 1$ for all $t \in [0,T]$, we obtain for each $i = 1,\ldots,n$
\begin{align}
\notag &\mathbb{E} \bigg[ \sup_{t \in [0,T]} \| \Psi_t^i \|^2 \bigg] = \mathbb{E} \left[ \sup_{t \in [0,T]} \bigg\| \int_0^t S_{t-s} \sigma_i(s,r_{s-}) dM_s^i \bigg\|^2 \right]
\\ \label{Psi-i-in-S2} &= \mathbb{E} \left[ \sup_{t \in [0,T]} \bigg\| e^{\omega t} \pi U_t \int_0^t e^{-\omega s} U_{-s} \ell \sigma_i(s,r_{s-}) dM_s^i \bigg\|^2 \right]
\\ \notag &\leq e^{2 \omega T} \mathbb{E} \left[ \sup_{t \in [0,T]} \bigg\| \int_0^t e^{- \omega s} U_{-s} \ell \sigma_i(s,r_{s-}) dM_s^i \bigg\|^2 \right] < \infty.
\end{align}
The latter expression is finite by Theorem \ref{thm-cadlag-mod} and Theorem \ref{thm-vanGaans-versus-usual}.

We obtain by (\ref{def-new-semigroup}), (\ref{diagram-commutes}), H\"older's inequality and Fubini's theorem (note that $\| \tilde{\alpha}(t,r_t) \|^2$ is c\`adl\`ag and therefore $\mathcal{B}[0,T] \otimes \mathcal{F}$-measurable)
\begin{align*}
&\mathbb{E} \bigg[ \sup_{t \in [0,T]} \| \Phi_t \|^2 \bigg] = \mathbb{E} \left[ \sup_{t \in [0,T]} \bigg\| \int_0^t S_{t-s} \tilde{\alpha}(s,r_s) ds \bigg\|^2 \right] 
\\ &= \mathbb{E} \left[ \sup_{t \in [0,T]} \bigg\| e^{\omega t} \pi U_t \int_0^t e^{-\omega s} U_{-s} \ell \tilde{\alpha}(s,r_s) ds \bigg\|^2 \right] \leq T e^{2\omega T} \mathbb{E} \bigg[ \int_0^T \| \tilde{\alpha}(t,r_t) \|^2 dt \bigg]
\\ &\leq T e^{2\omega T} \int_0^T \mathbb{E} \left[ \| \tilde{\alpha}(t,r_t) \|^2 \right] dt \leq T^2 e^{2\omega T} \sup_{t \in [0,T]} \mathbb{E} \left[ \| \tilde{\alpha}(t,r_t) \|^2 \right] < \infty.
\end{align*}
The latter supremum is finite, because $t \mapsto \mathbb{E} \left[ \| \tilde{\alpha}(t,r_t) \|^2 \right]$ is continuous on the compact interval $[0,T]$, as $t \mapsto \tilde{\alpha}(t,r_t)$ is continuous by the continuity of $r : [0,T] \rightarrow L^2(\Omega;H)$ and (\ref{lip-alpha-van-Gaans}), (\ref{lip-sigma-van-Gaans}).
Since the solution process $(r_t)$ is given by (\ref{short-solution}), we obtain, together with (\ref{Psi-i-in-S2}), that (\ref{solution-in-S2}) is valid.

We proceed by showing that $(r_t)_{t \geq 0}$ is also a weak solution to (\ref{SDE-to-solve}). Let $\zeta \in \mathcal{D}(A^*)$ be arbitrary.

We define for arbitrary $T \in \mathbb{R}_+$ the $\mathcal{B}[0,T] \otimes \mathcal{P}$-measurable functions $H_i : [0,T] \times [0,T] \times \Omega \rightarrow \mathbb{R}$ as
\begin{align*}
H_i(a,t) := \left\{
\begin{array}{cl}
\langle A^* \zeta, S_{a-t} \sigma_i(t,r_{t-}) \rangle, & a \geq t \medskip
\\ 0, & a < t.
\end{array}
\right.
\end{align*}
We obtain by the Cauchy-Schwarz inequality and the pseudo-contractivity of $(S_t)$ in $H_0$ that
\begin{align}\label{est-with-CS}
|H_i(a,t)| \leq e^{\omega T} \| A^* \zeta \| \cdot \| \sigma_i(t,r_{t-}) \|, \quad a,t \in [0,T].
\end{align}
The processes $Y_t^i := ( \int_0^T H_i^2(a,t)da )^{1/2}$ are left-continuous by (\ref{est-with-CS}) and Lebesgue's dominated convergence theorem, and therefore predictable. The L\'evy martingales $M^i$, considered on $[0,T]$, belong to $\mathcal{H}^2$ in the sense of the Definition in Protter \cite[p. 156]{Protter}, because, by using \cite[Thm. I.4.52]{JS},
\begin{align*}
\| M^i \|_{\mathcal{H}^2} &\leq \| M_i^c \|_{\mathcal{H}^2} + \| M_i^d \|_{\mathcal{H}^2} = \mathbb{E} \big[ [M_i^c,M_i^c]_T \big]^{1/2} + \mathbb{E} \big[ [M_i^d,M_i^d]_T \big]^{1/2}
\\ &= (c_i T)^{1/2} + \mathbb{E} \bigg[ \sum_{s \leq T} (\Delta M_s^i)^2 \bigg]^{1/2} = (c_i T)^{1/2} + \bigg( T \int_{\mathbb{R}} x^2 F_i(dx) \bigg)^{1/2} < \infty,
\end{align*}
where we have decomposed $M^i = M_i^c + M_i^d$ into its continuous and purely discontinuous martingale part, and where $c_i$ denotes the Gaussian part and $F_i$ the L\'evy measure of $M^i$.

There is, by the assumed continuity of $\sigma_1,\ldots,\sigma_n$, a constant $C_T > 0$ such that 
$\| \sigma_i(t,0) \| \leq C_T$ for all $t \in [0,T]$ and $i = 1,\ldots,n$. Therefore, we get for all $t \in [0,T]$, all $h \in H$ and all $i = 1,\ldots,n$ by (\ref{lip-sigma-van-Gaans})
\begin{align}
\label{linear-growth}
\| \sigma_i(t,h) \| &\leq \| \sigma_i(t,0) \| + \| \sigma_i(t,h) - \sigma_i(t,0) \| \leq (L \vee C_T) (1 + \| h \|).
\end{align}
By inequalities (\ref{est-with-CS}) and (\ref{linear-growth}) we obtain
\begin{align*}
&\mathbb{E} \bigg[ \int_0^T (Y_t^i)^2 d[M^i,M^i]_t \bigg] = \bigg( c_i + \int_{\mathbb{R}} x^2 F_i(dx) \bigg) \mathbb{E} \bigg[ \int_0^T \int_0^T H_i^2(a,t) da dt \bigg]
\\ &\leq T^2 (L \vee C_T)^2 e^{2 \omega T} \| A^* \zeta \|^2 \bigg( c_i + \int_{\mathbb{R}} x^2 F_i(dx) \bigg) \mathbb{E} \bigg[ \sup_{t \in [0,T]} (1 + \| r_t \|)^2 \bigg].
\end{align*}
Thus, the processes $Y^i$ are $(\mathcal{H}^2,M^i)$ integrable in the sense of the Definition in Protter \cite[p. 165]{Protter}, because from H\"older's inequality and (\ref{solution-in-S2}) we infer
\begin{align*}
\mathbb{E} \bigg[ \sup_{t \in [0,T]} (1 + \| r_t \|)^2 \bigg] \leq 1 + 2 \mathbb{E} \bigg[ \sup_{t \in [0,T]} \| r_t \|^2 \bigg]^{1/2} + \mathbb{E} \bigg[ \sup_{t \in [0,T]} \| r_t \|^2 \bigg] < \infty.
\end{align*}
Consequently, we have $Y^i \in L(M^i)$, that is each $Y^i$ is $M^i$ integrable in the sense of Protter \cite[p. 165]{Protter}, and therefore we may apply the Fubini Theorem, see Thm. IV.65 in \cite{Protter}, for the integrands $H_i$. Using the Fubini Theorem and \cite[Lemma VII.4.5(a)]{Werner}, we obtain for each $i = 1,\ldots,n$
\begin{align*}
&\int_0^t \langle A^* \zeta, \Psi_s^i \rangle ds = \int_0^t \int_0^s \langle A^* \zeta, S_{s-u} \sigma_i(u,r_{u-}) \rangle dM_u^i ds
\\ &= \int_0^t \Big\langle A^* \zeta, \int_u^t S_{s-u} \sigma_i(u,r_{u-}) ds \Big\rangle dM_u^i = \int_0^t \Big\langle \zeta, A \int_0^{t-u} S_s \sigma_i(u,r_{u-})ds \Big\rangle dM_u^i 
\\ &= \int_0^t \langle \zeta, S_{t-u} \sigma_i(u,r_{u-}) - \sigma_i(u,r_{u-}) \rangle dM_u^i = \langle \zeta, \Psi_t^i \rangle - \int_0^t \langle \zeta, \sigma_i(s,r_{s-}) \rangle dM_s^i,
\end{align*}
where the $\Psi^i$ are defined in (\ref{def-Psi-weak}). An analogous calculation, using the standard Fubini theorem, gives us
\begin{align*}
&\int_0^t \langle A^* \zeta, \Phi_s \rangle ds = \langle \zeta, \Phi_t \rangle - \int_0^t \langle \zeta, \tilde{\alpha}(s,r_s) \rangle ds,
\end{align*}
where $\Phi$ is defined in (\ref{def-Phi-weak}), and finally, we get, by taking into account \cite[Lemma VII.4.5(a)]{Werner} again,
\begin{align*}
\int_0^t \langle A^* \zeta, S_s h_0 \rangle ds = \Big\langle \zeta, A \int_0^t S_s h_0 ds \Big\rangle = \langle \zeta, S_t h_0 \rangle - \langle \zeta, h_0 \rangle.
\end{align*}
Together with (\ref{short-solution}), the latter three identities show that
\begin{align*}
\langle \zeta, r_t \rangle = \langle \zeta, h_0 \rangle + \int_0^t \Big( \langle A^* \zeta, r_s \rangle + \langle \zeta, \tilde{\alpha}(s,r_s) \rangle \Big) ds + \sum_{i=1}^n \int_0^t \langle \zeta, \sigma_i(s,r_{s-}) \rangle dM_s^i
\end{align*}
for all $t \in [0,T]$. Since $T \in \mathbb{R}_+$ was arbitrary, $(r_t)_{t \geq 0}$ is a weak solution to (\ref{SDE-to-solve}), as it fulfills (\ref{weak-solution}).

It remains to show that this weak solution is unique. Let $(r_t)_{t \geq 0}$ be any adapted c\`adl\`ag weak solution to (\ref{SDE-to-solve}), i.e. $(r_t)$ satisfies (\ref{weak-solution}) for all $\zeta \in \mathcal{D}(A^*)$. Let $\zeta \in \mathcal{D}(A^*)$ and $g \in C^1([0,T];\mathbb{R})$ for an arbitrary $T \in \mathbb{R}_+$. By the definition of the quadratic co-variation $[X,Y]$, see e.g. \cite[Def. I.4.45]{JS}, we obtain
\begin{align*}
\langle g(t) \zeta, r_t \rangle = \langle g(0) \zeta, h_0 \rangle + \int_0^t g(s) d \langle \zeta, r_s \rangle + \int_0^t \langle \zeta, r_s \rangle dg(s) + [g,\langle \zeta, r \rangle]_t.
\end{align*}
Since $g \in C^1([0,T];\mathbb{R})$, we have $[g,\langle \zeta, r \rangle] = 0$
according to \cite[Prop. 4.49.d]{JS}. Therefore and because of (\ref{weak-solution}), we get
\begin{align*}
\langle g(t) \zeta, r_t \rangle &= \langle g(0) \zeta, h_0 \rangle + \int_0^t \Big( \langle g'(s) \zeta + A^* g(s) \zeta, r_s \rangle + \langle g(s) \zeta, \alpha(s,r_s) \rangle \Big) ds 
\\ & \quad + \sum_{i=1}^n \int_0^t \langle g(s) \zeta, \sigma_i(s,r_{s-}) \rangle dX_s^i.
\end{align*}
Since the set $\{ t \mapsto g(t) \zeta \, | \, g \in C^1([0,T];\mathbb{R}) \}$ is dense in $C^1([0,T];\mathcal{D}(A^*))$, we deduce
\begin{align*}
\langle g(t), r_t \rangle &= \langle g(0), h_0 \rangle + \int_0^t \Big( \langle g'(s) + A^* g(s), r_s \rangle + \langle g(s), \alpha(s,r_s) \rangle \Big) ds 
\\ & \quad + \sum_{i=1}^n \int_0^t \langle g(s), \sigma_i(s,r_{s-}) \rangle dX_s^i
\end{align*}
for all $g \in C^1([0,T];\mathcal{D}(A^*))$, where we recall that $T \in \mathbb{R}_+$ was arbitrary. Defining $g \in C^1([0,t];\mathcal{D}(A^*))$ for an arbitrary 
$t \in \mathbb{R}_+$ and an arbitrary $\zeta \in \mathcal{D}(A^*)$ as
$g(s) := S_{t-s}^* \zeta$, $s \in [0,t]$, we obtain $g'(s) = -A^* g(s)$, and hence
\begin{align*}
\langle \zeta, r_t \rangle = \langle \zeta, S_t h_0 \rangle + \int_0^t \langle \zeta, S_{t-s} \alpha(s,r_s) \rangle ds + \sum_{i=1}^n \int_0^t \langle \zeta, S_{t-s} \sigma_i(s,r_{s-}) \rangle dX_s^i.
\end{align*}
Since $\mathcal{D}(A^*)$ is dense in $H$, the process $(r_t)_{t \geq 0}$ is also a mild solution to (\ref{SDE-to-solve}), i.e. it satisfies (\ref{mild-solution}), proving the desired uniqueness.
\end{proof}

In the special situation where $A \in \mathcal{L}(H)$, i.e. $A$ is a bounded linear operator, we can now easily establish the existence of a strong solution $(r_t)_{t \geq 0}$ to (\ref{SDE-to-solve}), that is we have
\begin{align}\label{strong-solution}
r_t = h_0 + \int_0^t \Big( A r_s + \alpha(s,r_s) \Big) ds + \sum_{i=1}^n \int_0^t \sigma_i(s,r_{s-}) dX_s^i, \quad t \geq 0.
\end{align}

\begin{corollary}\label{cor-solution-of-SDE}
Let $A \in \mathcal{L}(H)$ be a bounded linear operator and let $\alpha,\sigma_1,\ldots,\sigma_n : \mathbb{R}_+ \times H \rightarrow H$ be continuous. Assume there is constant $L \geq 0$ such that (\ref{lip-alpha-van-Gaans}) and (\ref{lip-sigma-van-Gaans}) are satisfied for all $t \in \mathbb{R}_+$ and $h_1, h_2 \in H$. Then, for each $h_0 \in H$, there exists a unique strong adapted c\`adl\`ag solution $(r_t)_{t \geq 0}$ to (\ref{SDE-to-solve}) with $r_0 = h_0$ satisfying (\ref{solution-in-S2}).
\end{corollary}

\begin{proof}
The operator $A$ is generated by the semigroup $S_t = e^{tA}$, which is pseudo-contractive, because
\begin{align*}
\| S_t \|_{\mathcal{L}(H)} \leq e^{t \| A \|_{\mathcal{L}(H)}}, \quad t \geq 0.
\end{align*}
By Theorem \ref{thm-sde-solution}, for each $h_0 \in H$, there exists a unique weak adapted c\`adl\`ag solution $(r_t)_{t \geq 0}$ to (\ref{SDE-to-solve}) with $r_0 = h_0$ satisfying (\ref{solution-in-S2}), which also fulfills (\ref{strong-solution}) by the boundedness of $A$, showing that $(r_t)$ is a strong solution to (\ref{SDE-to-solve}). 
\end{proof}

We close this section with a couple of remarks. Actually, \cite[Thm. 4.1]{Onno-Levy} is not explicitly proven in \cite{Onno-Levy}. We quote \cite[p. 19]{Onno-Levy}: "For a proof of Theorem 4.1 one can follow almost literally the proofs of Theorem 4.1 and Theorem 4.2 in \cite{Onno}, $\ldots$". The mentioned result, \cite[Thm. 4.1]{Onno}, is
 an analogous result for stochastic equations driven by an infinite dimensional Brownian motion.

Note that the existence result of van Gaans \cite[Thm. 4.1]{Onno-Levy} demands no further assumptions on the $C_0$-semigroup. In contrast, we require the pseudo-contractivity of $(S_t)$ in a closed subspace in order to prove that the solution possesses a c\`adl\`ag modification.

The idea to use the Szek\"ofalvi-Nagy's theorem on unitary dilations in order to overcome the difficulties arising from stochastic convolutions, is due to Hausenblas and Seidler, see \cite{Seidler} and \cite{Seidler2}.

Without using the Szek\"ofalvi-Nagy's theorem, Baudoin and Teichmann \cite{Teichmann-groups} consider stochastic equations on separable Hilbert spaces equipped with a strongly continuous group, in Sec. 3 of their article also with focus on interest rate theory.

For every pseudo-contractive semigroup $(S_t)$, stochastic convolutions $\int_0^t S_{t-s} \Phi_s dM_s$ with respect to a square-integrable, c\`adl\`ag martingale $M$ have a c\`adl\`ag modification, which is due to Kotelenez \cite{Kotelenez}. We use the Szek\"ofalvi-Nagy's theorem on unitary dilations in order to get a c\`adl\`ag modification, because we deal with the stochastic integral $\text{{\rm (G-)}}\int_0^t S_{t-s} \Phi_s dM_s$ defined in van Gaans \cite[Sec. 3]{Onno-Levy}.

Recently, there has been growing interest in stochastic differential equations of the type (\ref{SDE-to-solve}) with jump noise terms. As a result, a few related papers \cite{Ruediger-mild,Knoche1,Knoche2, Hausenblas,Hausenblas2,Thilo,P-Z-paper} and the forthcoming textbook \cite{P-Z-book} have been written, but mostly with other fields of applications than finance.

During the revision of this paper we became aware of the recent preprint \cite{P-Z-paper}, where the authors derived independently similar results. But they work on different function spaces where the forward curve is not necessarily continuous and thus the short rate is not well defined. Moreover, they only consider volatilities of composition type, that is $\sigma_i(t,r)(x) = g_i(t,x,r(x))$ with deterministic functions $g_i : \mathbb{R}_+ \times \mathbb{R}_+ \times \mathbb{R} \rightarrow \mathbb{R}$.

\end{appendix}


\begin{thebibliography}{20}

\bibitem{Ruediger-mild} Albeverio, S., Mandrekar, V., R\"udiger, B. (2006):
  Existence of mild solutions for stochastic differential equations and semilinear equations with non Gaussian L\'evy noise.
  Preprint no. 314, SFB 611, University of Bonn.
  
  \bibitem{Barndorff-Nielsen} Barndorff--Nielsen, O.~E. (1977)
  Exponentially decreasing distributions for the logarithm of particle size.
  \textit{Proceedings of the Royal Society London} Series A, Vol. 353, 401--419.
  
  \bibitem{Teichmann-groups} Baudoin, F., Teichmann, J. (2005):
  Hypoellipticity in infinite dimensions and an application to interest rate theory.
  \textit{Annals of Applied Probability} {\bf 15}(3), 1765--1777.
  
  \bibitem{Bhar} Bhar, R., Chiarella, C. (1997):
  Transformation of Heath--Jarrow--Morton models to Markovian systems.
  \textit{The European Journal of Finance} {\bf 3}, 1--26.
  
  \bibitem{BKR} Bj{\"o}rk, T., Di Masi, G., Kabanov, Y., Runggaldier, W. (1997):
  Towards a general theory of bond markets.
  \textit{Finance and Stochastics} {\bf 1}(2), 141--174.
  
  \bibitem{BKR0} 
  Bj{\"o}rk, T., Kabanov, Y., Runggaldier, W. (1997):
  Bond market structure in the presence of marked point processes.
  \textit{Mathematical Finance} {\bf 7}(2), 211--239. 
  
  \bibitem{Bj_Sv} Bj{\"o}rk, T., Svensson, L. (2001):
  On the existence of finite dimensional realizations for nonlinear forward rate models.
  \textit{Mathematical Finance} {\bf 11}(2), 205--243.

\bibitem{carteh} Carmona, R., Tehranchi, M. (2006): \textit{Interest rate models: an infinite dimensional stochastic analysis perspective.} Berlin: Springer.

\bibitem{CGMY} Carr, P., Geman, H., Madan, D., Yor, M. (2002)
    The fine structure of asset returns: an empirical investigation.
    \textit{Journal of Business} {\bf 75}(2), 305-332.

\bibitem{Cont-Tankov}
    Cont, R., Tankov, P. (2004)    
    \emph {Financial modelling with jump processes.}
    Chapman and Hall / CRC Press, London.    
    
  \bibitem{Kwon_2001} Chiarella, C., Kwon, O.~K. (2001):
  Forward rate dependent Markovian transformations of the Heath--Jarrow--Morton term structure model.
  \textit{Finance and Stochastics} {\bf 5}(2), 237--257.
  
  \bibitem{Kwon_2003} Chiarella, C., Kwon, O.~K. (2003):
  Finite dimensional affine realizations of HJM models in terms of forward rates and yields.
  \textit{Review of Derivatives Research} {\bf 6}(3), 129--155.
  
  \bibitem{Da_Prato} Da~Prato, G., Zabczyk, J. (1992):
  \textit{Stochastic equations in infinite dimensions.} New York: Cambridge University Press.

  \bibitem{Davies} Davies, E.~B. (1976):
  \textit{Quantum theory of open systems.} London: Academic Press.
  
  \bibitem{ds94} Delbaen, F., Schachermayer, W. (1994): A general version of the fundamental theorem of asset pricing.
  \textit{Mathematische Annalen} {\bf 300}, 463--520.
  
\bibitem{Eberlein_J} Eberlein, E., Jacod, J., Raible, S. (2005):
  L\'evy term structure models: no-arbitrage and completeness.
  \textit{Finance and Stochastics} {\bf 9}, 67--88.
  
  \bibitem{Eberlein-Keller} Eberlein, E., Keller, U. (1995):
  Hyperbolic distributions in finance.
  \textit{Bernoulli} {\bf 1}, 281--299.
  
  \bibitem{Eberlein_Kluge_Review} 
    Eberlein, E., Kluge, W. (2007):
    Calibration of L\'evy term structure models.
    In \textit{Advances in Mathematical Finance:} In Honor of Dilip Madan, M. Fu, R.~A. Jarrow, J.-Y. Yen, and R.~J. Elliott (Eds.), Birkh\"auser, pp. 155--180. 
  
  \bibitem{Eberlein_K1} Eberlein, E., Kluge, W. (2006):
  Exact pricing formulae for caps and swaptions in a L\'evy term structure model.
  \textit{Journal of Computational Finance} {\bf 9}(2), 99--125.
  
  \bibitem{Eberlein_K2} Eberlein, E., Kluge, W. (2006):
  Valuation of floating range notes in L\'evy term structure models.
  \textit{Mathematical Finance} {\bf 16}, 237--254.
  
  \bibitem{Eberlein_O} 
    Eberlein, E., \"Ozkan, F. (2003)
    The defaultable L\'evy term structure: ratings and restructuring.
    \textit{Mathematical Finance} {\bf 13}, 277--300.    
  
  \bibitem{Eberlein-Raible} Eberlein, E., Raible, S. (1999):
  Term structure models driven by general L\'evy processes.
  \textit{Mathematical Finance} {\bf 9}(1), 31--53.
  
  \bibitem{Engel-Nagel} Engel, K.-J., Nagel, R. (2000):
  \textit{One-parameter semigroups for linear evolution equations.} New York: Springer.
  
  
\bibitem{fillnm} Filipovi\'c, D. (2001): \textit{Consistency problems for Heath--Jarrow--Morton interest rate models.} Berlin: Springer.

  

  \bibitem{Filipovic} Filipovi\'c, D., Teichmann, J. (2003):
  Existence of invariant manifolds for stochastic equations in infinite dimension.
  \textit{Journal of Functional Analysis} {\bf 197}, 398--432.
  
  \bibitem{Onno} van Gaans, O. (2005):
  A series approach to stochastic differential equations with infinite dimensional noise.
  \textit{Integral Equations and Operator Theory} {\bf 51}(3), 435--458. 
  
\bibitem{Onno-Levy} van Gaans, O. (2005):
  Invariant measures for stochastic evolution equations with L\'evy noise.
  Technical Report, Leiden University. 
  {\tt (www.math.leidenuniv.nl/$\sim$vangaans/publications.html)}  
  
  \bibitem{Hausenblas} Hausenblas, E. (2005):
  SPDEs driven by Poisson random measure: Existence and uniqueness.
  \textit{Electronic Journal of Probability} {\bf 11}, 1496--1546.
  
  \bibitem{Hausenblas2} Hausenblas, E. (2007):
  SPDEs driven by Poisson random measure with non Lipschitz coefficients: Existence results.
  Forthcoming in \textit{Probability Theory and Related Fields.}
  
  \bibitem{Seidler2} Hausenblas, E., Seidler, J. (2001):
  A note on maximal inequality for stochastic convolutions.
  \textit{Czechoslovak Mathematical Journal} {\bf 51}(126), 785--790.
  
  \bibitem{Seidler} Hausenblas, E., Seidler, J. (2007):
  Stochastic convolutions driven by martingales: Maximal inequalities and exponential integrability.
  forthcoming in \textit{Stoch. Anal. Appl.}
  
  \bibitem{HJM} Heath, D., Jarrow, R., Morton, A. (1992):
  Bond pricing and the term structure of interest rates: a new methodology for contingent claims valuation.
  \textit{Econometrica} {\bf 60}(1), 77--105.
  
  \bibitem{Hyll} Hyll, M. (2000):
  Affine term structures and short-rate realizations of forward rate models driven by jump-diffusion processes.
  In \textit{Essays on the term structure of interest rates} PhD thesis, Stockholm School of Economics.
  
  \bibitem{Inui} Inui, K., Kijima, M. (1998):
  A Markovian framework in multi-factor Heath--Jarrow--Morton models.
  \textit{Journal of Financial and Quantitative Analysis} {\bf 33}(3), 423--440.
  
  \bibitem{JS} Jacod, J., Shiryaev, A.~N. (1987):
  \textit{Limit theorems for stochastic processes.} Berlin: Springer.
  
  
  \bibitem{Zabcyk_F2} Jakubowski, J., Zabczyk, J. (2006):
  Exponential moments for HJM models with jumps.
  Preprint IMPAN 673, Warsaw.
  
  
  \bibitem{Jarrow_Madan}
    Jarrow, A., Madan, D.~B. (1995)
    Option pricing using the term structure of interest rates to hedge systematic discontinuities in asset returns.
    \textit{Mathematical Finance} {\bf 5}(4), 311--336. 
  
  \bibitem{Jeffrey} Jeffrey, A. (1995):
  Single factor Heath--Jarrow--Morton term structure models based on Markov spot interest rate dynamics.
  \textit{Journal of Financial and Quantitative Analysis} {\bf 30}(4), 619--642.
  
  \bibitem{Knoche1} Knoche, C. (2004):
  SPDEs in infinite dimensions with Poisson noise. Comptes Rendus Math\'ematique. 
  Acad\'emie des Sciences. Paris, Serie I 339, 647--652.
  
  
  \bibitem{Knoche2} Knoche, C. (2005):
  Mild solutions of SPDEs driven by Poisson noise in infinite dimensions and their dependence on initial conditions.
  PhD thesis, University of Bielefeld.
  
  
  \bibitem{Kotelenez} Kotelenez, P. (1982):
  A submartingale type inequality with applications to stochastic evolution equations.
  \textit{Stochastics} {\bf 8}, 139--151.
  
  \bibitem{Kuechler-Tappe} K\"uchler, U., Tappe, S. (2007):
  Bilateral Gamma distributions and processes in financial mathematics.
  Forthcoming in \textit{Stochastic Processes and their Applications.}
  
  
  \bibitem{Madan} Madan, D.~B. (2001)
  Purely discontinuous asset pricing processes.
  In: Jouini, E., Cvitani{\v c}, J. and Musiela, M. (Eds.), pp. 105--153
  \textit{Option Pricing, Interest Rates and Risk Management.}
  Cambridge University Press, Cambridge.
  
  
  \bibitem{Metivier} M\'etivier, M. (1982):
  \textit{Semimartingales.} Walter de Gruyter, Berlin.
  
  \bibitem{Thilo} Meyer-Brandis, T. (2005):
  Differential equations driven by L\'evy white noise in spaces of Hilbert space valued stochastic distributions.
  Preprint, University of Oslo.   
  {\tt (www.math.uio.no/eprint/pure\_math/2005/09-05.pdf)}
  
  \bibitem{Musiela} Musiela, M. (1993):
  Stochastic PDEs and term structure models.
  \textit{Journ\'ees Internationales de Finance}, IGR-AFFI, La Baule.
  
  \bibitem{Schmidt} \"Ozkan, F., Schmidt, T. (2005):
  Credit risk with infinite dimensional L\'evy processes.
  \textit{Statistics and Decisions} {\bf 23}, 281--299.
  
  \bibitem{Protter} Protter, P. (2005):
  \textit{Stochastic integration and differential equations.} Second Edition, Version 2.1, Berlin: Springer.
  
  
  \bibitem{P-Z-book} Peszat, S., Zabczyk, J. (2007):
  \textit{Stochastic partial differential equations with L\'evy noise: Evolution equations approach.}
  Cambridge University Press, Cambridge. \textit{to appear}.
  
  \bibitem{P-Z-paper} Peszat, S., Zabczyk, J. (2007):
  Heath-Jarrow-Morton-Musiela equation of bond market.
  Preprint IMPAN 677, Warsaw.
  {\tt (www.impan.gov.pl/EN/Preprints/index.html)}
  
  \bibitem{Raible} Raible, S. (2000):
  L\'evy processes in finance: theory, numerics, and empirical facts.
  PhD thesis, University of Freiburg.
  
  \bibitem{Ritchken} Ritchken, P., Sankarasubramanian, L. (1995):
  Volatility structures of forward rates and the dynamics of the term structure.
  \textit{Mathematical Finance} {\bf 5}(1), 55--72.
  
  \bibitem{barbara-integration} R\"udiger, B. (2004):
  Stochastic integration with respect to compensated Poisson random measures on separable Banach spaces.
  \textit{Stoch. Stoch. Rep.} {\bf 76}(3), 213--242.
  
  \bibitem{Sato} Sato, K. (1999):
  \textit{L\'evy processes and infinitely divisible distributions.}
  Cambridge studies in advanced mathematics, Cambridge.
  
  \bibitem{Shirakawa}
    Shirakawa, H. (1991)
    Interest rate option pricing with Poisson-Gaussian forward rate curve processes.
    \textit{Mathematical Finance} {\bf 1}(4), 77--94.
  
  \bibitem{Nagy} Sz.-Nagy, B., Foia\c{s}, C. (1970):
  \textit{Harmonic analysis of operators on Hilbert space.} North-Holland, Amsterdam.
  
  
  
  \bibitem{Tehranchi} Tehranchi, M. (2005):
  A note on invariant measures for HJM models.
  \textit{Finance and Stochastics} {\bf 9}(3), 389--398.
  
  \bibitem{Werner} Werner, D. (2002):
  \textit{Funktionalanalysis.} Berlin: Springer.
  
  

\end{thebibliography}
\end{document}